\documentclass[12pt]{article}
\newcommand{\blind}{0}
\usepackage{todonotes}
\usepackage{algorithm}
\usepackage{algorithmic}
\usepackage{float}
\usepackage{hyperref}
\usepackage{paralist, amsmath, amsthm, amssymb, color, graphicx}
\usepackage[margin=1in]{geometry}
\usepackage{verbatim}
\usepackage{amssymb}
\usepackage{setspace}
\usepackage{mathtools}
\usepackage{color}
\usepackage{cancel}

\usepackage[round]{natbib}

\usepackage{stmaryrd}

\usepackage{xfrac}

\usepackage{mathrsfs}
\usepackage{amscd}
\usepackage{dsfont}
\usepackage{tikz}
\usetikzlibrary{arrows}
\usetikzlibrary{graphs}
\usetikzlibrary{arrows.meta}

\usepackage{enumerate}

\theoremstyle{plain}
\newtheorem{thm}{Theorem}[section] 
\newtheorem{lem}[thm]{Lemma}
\newtheorem{prop}[thm]{Proposition}
\newtheorem{cor}[thm]{Corollary}

\theoremstyle{definition}
\newtheorem{defn}[thm]{Definition} 

\newtheorem{example}[thm]{Example}
\newtheorem{remark}[thm]{Remark}

\usepackage{listings}
\usepackage{color} 
\definecolor{mygreen}{RGB}{28,172,0} 
\definecolor{mylilas}{RGB}{170,55,241}
\definecolor{mygray}{gray}{0.95}

\newcommand{\mscr}[1]{\mathscr{#1}}

\newcommand{\RR}{\mathbb{R}}

\newcommand{\EE}{\mathbb{E}}

\newcommand{\al}{\alpha}

\newcommand{\ga}{\gamma}
\newcommand{\de}{\delta}
\newcommand{\ep}{\epsilon}

\newcommand{\ta}{\theta}
\newcommand{\Ta}{\Theta}
\newcommand{\la}{\lambda}

\newcommand{\sas}{\sigma^2}

\newcommand{\De}{\Delta}

\usepackage{mathtools}

\renewcommand{\l}{\left}
\renewcommand{\r}{\right}

\newcommand{\defeq}{\vcentcolon=}

\newcommand{\iid}{\overset{\text{iid}}{\sim}}

\DeclareMathOperator{\var}{Var}

\newcommand{\ds}{\displaystyle}

\newcommand{\indep}{\rotatebox[origin=c]{90}{$\models$}}

\newcommand{\fig}[3]{\begin{figure}[h!]\begin{center}\includegraphics[#1]{#2}\end{center}\caption{#3}\label{fig:#2}\end{figure}}

\title{Structure and Sensitivity in Differential Privacy:\\ Comparing K-Norm Mechanisms}

\date{}

\graphicspath{{./Pictures/}}
\begin{document}

\def\spacingset#1{\renewcommand{\baselinestretch}%
{#1}\small\normalsize} \spacingset{1}

\if1\blind
{
  \author{}
  \maketitle
} \fi

\if0\blind
{
\author{Jordan Awan \& Aleksandra Slavkovi\'c}
\maketitle
  \medskip
} \fi

\bigskip
\begin{abstract}
Differential privacy (DP), provides a framework for provable privacy protection against arbitrary adversaries, while allowing the release of summary statistics and synthetic data. We address the problem of releasing a noisy real-valued statistic vector $T$, a function of sensitive data under DP, via the class of $K$-norm mechanisms with the goal of minimizing the noise added to achieve privacy. First, we introduce the \emph{sensitivity space of $T$}, which extends the concepts of sensitivity polytope and sensitivity hull to the setting of arbitrary statistics $T$. We then propose a framework consisting of three methods for comparing the $K$-norm mechanisms: 1) a multivariate extension of stochastic dominance, 2) the entropy of the mechanism, and 3) the conditional variance given a direction, to identify the optimal $K$-norm mechanism. In all of these criteria, the optimal $K$-norm mechanism is generated by the convex hull of the sensitivity space. Using our methodology, we extend the objective perturbation and functional mechanisms and apply these tools to logistic and linear regression, allowing for private releases of statistical results. Via simulations and an application to a housing price dataset, we demonstrate that our proposed methodology offers a substantial improvement in utility for the same level of risk. 
\end{abstract}

\noindent%
\vfill

\newpage


%

%


\section{Introduction}
\label{Introduction}

Statistical Disclosure Limitation (SDL) or Control (SDC) refers to the broad class of methods developed to address the trade-off between limiting the disclosure risk of sharing and publishing sensitive data, while at the same time maintaining the utility and validity of statistical inference \citep{Duncan1986,Duncan1989,Fienberg1998}. Until recently, the statistical literature on disclosure limitation has built on a probabilistic notion of disclosure as proposed by   \citet{Dalenius1977}: ``If the release of the statistics [T] makes it possible to determine the value [of confidential statistical data] more accurately than is possible without access to [T], a disclosure has taken place.'' However, while many statistical procedures appear innocuous, \citet{Dwork2017:Exposed} demonstrate through a survey of potential disclosure risks (i.e., possible attacks on private data) due to release of aggregate statistics, that it can be nontrivial to determine if a disclosure has taken place or not. In fact, while Dalenius' notion of disclosure risk is intuitive, formalizing this definition requires knowing the prior knowledge of the attacker, their side information, as well as their computational power. In contrast, differential privacy (DP), introduced in \citep{Dwork2006:Sensitivity}, has emerged as a formal framework to quantify the ``privacy level'', which provides provable protection against adversaries with arbitrary priors, unlimited side information, and unbounded computational power. Differential privacy guarantees that whether an individual is in a database or not, the results of a DP procedure should be similar in terms of their probability distribution; this guarantee offers a sense of ``plausible deniability'' and limits the ability of an adversary to infer about any particular individual in the database.  The strength of the privacy guarantee is characterized by a real value $\ep>0$, called the privacy-loss budget, where smaller values of $\ep$ provide a stronger privacy guarantee.

Differential privacy was proposed by the computer science community, but over the last decade an emphasis has been put on linking DP to fundamental statistical concepts in order to expand and improve its applicability. 
\citet{Wasserman2010:StatisticalFDP} showed that satisfying DP for a small value of $\ep$ 
guarantees that certain hypotheses an adversary may attempt to test about particular individuals in the dataset have low power. \citet{Wasserman2010:StatisticalFDP} also proposed tools for releasing histograms and density estimates under DP. In \citet{Dwork2009:Robust}, it was shown that the amount of noise required to privatize a statistic is related to concepts in the field of robust statistics. Since its inception, the DP framework has been expanded to include mechanisms for private releases of principal components \citep{Chaudhuri2013}, model selection \citep{Lei2016}, hypothesis tests \citep{Vu2009,Wang2015:Revisiting,Gaboardi2016,Awan2018:Binomial,Canonne2019}, confidence intervals \citep{Karwa2017}, network analysis \citep{Karwa2016:Sharing,Karwa2016}, as well as analyses of functional data \citep{Hall2013,Mirshani2019,Awan2019}, to list a few. Besides limiting privacy risk, \citet{Dwork2015:ADA} show that the tools of DP can be used to obtain accurate statistical inferences in adaptive data analysis.

However, DP is often criticized for a substantial drop in statistical utility and thus a lack of applicability, especially in finite sample settings. \citet{Bun2018} show that there is a significant sample complexity cost to answering a set of queries under DP, in terms of the dimension of the database. On the other hand, holding the dimension of the individual's data fixed, \citet{Smith2011:Privacy-preservingSE} shows that under mild assumptions, asymptotically efficient estimators are achievable under DP.  While it is encouraging that asymptotically, DP mechanisms can perform as well as non-private methods, practical implementations of DP mechanisms often suffer from a considerable drop in performance even for moderate sample sizes. For example, \citet{Fienberg2010} criticized DP mechanisms that release privatized contingency tables for producing unacceptably inaccurate results, and \citet{Vu2009} show that significant adjustments in the sample size and the sampling distribution of the private test statistic are required to conduct the simplest of hypothesis tests. Therefore, there is a need to not only develop DP tools for more statistical problems, but to optimize existing DP tools to improve their practical accuracy on finite sample data problems. 
%
In this paper, we address the problem of releasing a noisy real-valued statistic vector $T$, a function of sensitive data under DP, via the class of $K$-norm mechanisms with the goal of minimizing the noise added to achieve privacy, and optimizing the use of the privacy-loss budget $\ep$ for a fixed statistic and sample size $n$.

While the computer science community has produced finite sample complexity bounds for many problems (i.e. \citealp{Hardt2010:GeometryDP,Bun2018,Cai2017,Acharya2018,Canonne2019} for linear queries, identity testing, and hypothesis testing, respectively, to name a few), there is a much more limited literature on optimizing DP mechanisms for a fixed sample size. \citet{Ghosh2012} show that for any count query, the \emph{geometric mechanism} simultaneously minimizes the expected value of a wide variety of loss functions. \citet{Geng2015} developed a \emph{staircase mechanism} which they show maximizes utility for real-valued statistics, with a focus on $\ell_1$ and $\ell_2$ error. \citet{Geng2013} show that for two-dimensional real-valued statistics, a correlated multidimensional staircase mechanism minimizes the expected $\ell_1$ loss. \citet{Wang2014} show that for a database of one person, the minimum-entropy mechanism (with a differentiable density) is Laplace. \citet{Wang2017} shows that for $d$-dimensional databases with sensitivity based on the $\ell_1$ norm, the minimum-entropy mechanism is to add independent Laplace noise to each coordinate. \citet{Awan2018:Binomial} developed uniformly most powerful DP hypothesis tests for Bernoulli data based on the \emph{Tulap distribution}, a variant of the staircase mechanism. Furthermore, \citet{Karwa2016:Sharing} provide a method of obtaining maximum likelihood estimates for exponential graph models, which correct for the bias introduced by standard DP methods.


The original and most common method of achieving DP for the release of an $m$-dimensional real statistic vector $T$ is through the Laplace mechanism, which adds independent Laplace random variables to each entry of $T$ before releasing. The Laplace mechanism can be generalized to the K-norm mechanisms introduced by \citet{Hardt2010:GeometryDP}, a family of unbiased mechanisms (mean is the non-private $T$) which are the focus of this paper. The $K$-norm mechanisms (here on abbreviated as $K$-mech) are a family of additive mechanisms determined by the choice of a norm on $\RR^m$.  
Using the $\ell_1$ norm leads to the Laplace mechanism which is popularly used (e.g., see \citet{Dwork2006:Sensitivity,Smith2011:Privacy-preservingSE,Zhang2012:FunctionalMR,Yu2014}). In  \citet{Chaudhuri2009:Logistic}, \citet{Chaudhuri2011:DPERM}, \citet{Kifer2012:PrivateCERM}, \citet{Song2013:StochasticGradient}, and \citet{Yu2014}, the $\ell_2$ norm variant of the $K$-mech is used in applications of empirical risk minimization. The $\ell_\infty$ norm variant of the $K$-mech is proposed in \citet{Steinke2017}, as a mechanism to optimize the worst case error when releasing one-way marginals of a high-dimensional binary database.  \citet{Hall2012:NewSADP} developed a minimax-optimal $K$-mech for density estimation, and \citet{Xiao2015} study and implement $K$-mechs related to two-dimensional polytopes for the application of location data. With such a large class of mechanisms, it is natural to ask which $K$-mech maximizes the utility of the output. \citet{Hardt2010:GeometryDP} began this comparison by studying $K$-mechs 
to release linear statistics, with the goal of optimizing the sample complexity, as the sample size and dimension of $T$ increases. However, the restriction to linear statistics limits the types of problems that can be tackled, and simplifies the geometric problem. In contrast, our goal is to choose the best $K$-mech to optimize performance for an arbitrary fixed statistic $T$ and sample size, optimizing finite sample statistical utility. By allowing non-linear statistics, we are able to accommodate a larger set of applications; this extension to non-linear statistics also results in more complex geometric structures, which we explore.

{\bf Our Contributions } 
The objective of this paper is to optimize the performance of the $K$-norm mechanisms, for a given arbitrary statistic at a fixed sample size. To do this, we first introduce the new geometric notion, \emph{sensitivity space} $S_T$ of $T$, related to the \emph{sensitivity} of $T$ (i.e., the amount that $T$ can vary by changing one person's information in the dataset). When releasing a noisy version of a statistic $T$ under $\ep$-DP, the amount of noise required is related to the sensitivity of $T$ and 
an imprecise estimate of the sensitivity, which does not carefully consider the structure of $T$, can amplify the loss in statistical utility. 
The proposed sensitivity space allows for the rigorous theoretical and practical comparison of $K$-Mechs, with the goal of minimizing the magnitude of noise introduced to satisfy $\ep$-DP. 
It also generalizes the polytopes used in \citet{Hardt2010:GeometryDP}, which are called \emph{sensitivity polytope} in \citet{Dwork2014,Dwork2015,Nikolov2015,Kattis2016}, and \emph{sensitivity hull} in \citet{Xiao2015,Xiao2017}, all of which are designed for linear statistics. When extended to nonlinear statistics however, the sensitivity space need not be a polytope, and can take a wider variety of forms, which are mathematically more complex. Furthermore, the generalization to arbitrary statistics allows the $K$-mechs to be implemented within many existing DP mechanisms, and by choosing the $K$-mech carefully, we can significantly improve the performance of these mechanisms for many statistical problems, such as regression and empirical risk minimization.

In order to identify the optimal $K$-norm mechanism for a fixed statistic and sample size, we propose a novel framework of comparing the mechanisms that relies on measuring the sensitivity space. The framework consists of proposing three theoretical perspectives of comparing the $K$-norm mechanisms and linking them with two optimal decision rules. The three perspectives, (1) a multivariate version of stochastic dominance, (2) the entropy of the mechanism, and (3) the conditional variance given a unit direction, 
each result in one of two stochastic orderings on the K-mechs, which we call the \emph{containment order} and the \emph{volume order}. We show that in all of these criteria, the optimal K-norm mechanism is generated by the convex hull of our proposed sensitivity space, a fundamental result with which we generalize previous results in the literature such as the use of the $\ell_\infty$ norm in \citet{Steinke2017}). 
 The first method of comparison is a novel stochastic ordering which is a multivariate extension of stochastic dominance \citep{Quirk1962}, motivated by notions of statistical depth \citep{Mosler2013}, which could be of interest outside of privacy as an alternative stochastic ordering on distributions \citep{Zuo2000Nonparam}. 
Second, we compare the $K$-mechs in terms of the entropy of the noise-adding distribution. We show that the entropy of a $K$-mech is determined by the volume of its corresponding norm ball. 
Finally, we compare $K$-mechs in terms of their conditional variance given a unit direction, which we show is based on the containment of the norm balls. This work offers both geometric insight into the mechanisms as well as a simple decision criteria to choose between mechanisms.

Using our new methodology, 
we also extend the commonly used objective perturbation 
and functional mechanisms 
to permit arbitrary $K$-mechs, allowing for the application of our techniques to optimize the finite sample performance of these mechanisms.
Objective perturbation \citep{Chaudhuri2009:Logistic,Chaudhuri2011:DPERM,Kifer2012:PrivateCERM} and functional mechanism \citep{Zhang2012:FunctionalMR} are highly influential works  which have had a large impact on the field of differential privacy, and are among the state of the art mechanisms for regression problems\footnote{In our preliminary simulations on DP regression mechanisms, we found that objective perturbation and functional mechanism were among the top performing algorithms for logistic and linear regression, respectively, especially for moderate sample sizes.}. 
We illustrate how our theoretical methodology applies to the problems of logistic and linear regression via objective perturbation and functional mechanism, respectively. Through simulations and a real data application, we demonstrate that by carefully choosing the $K$-mech at crucial steps, we obtain significant gains in the finite-sample accuracy of these mechanisms for the same level of $\epsilon$, improving the applicability of these mechanisms for real data problems. 
From another perspective, by optimizing the performance we are able to provide the same level of accuracy with a smaller value of $\epsilon$, allowing for better use of the privacy-loss budget to answer other potential statistical queries.


{\bf Organization} In Section \ref{s:background}, we review the background of DP, introduce the sensitivity space, and demonstrate its connection to the $K$-norm mechanisms. We explore an example in Section \ref{s:explore} to demonstrate the mathematical structure of the sensitivity space for non-linear statistics. In the main section, Section \ref{s:compare}, we propose three methods to compare the $K$-norm mechanisms. In Subsection \ref{s:stochastic} we propose a stochastic partial ordering of the $K$-mechs, which in Subsection \ref{s:depth} we show is connected to concepts in statistical depth. In Subsection \ref{s:entropy} we derive the entropy of a $K$-mech in terms of the volume of its corresponding norm ball. In Subsection \ref{s:variance} we prove that the conditional variance of the mechanism is optimized based on the containment of the norm balls. In Subsection \ref{s:variance}, we show that under each of our criteria, the optimal mechanism is generated by the convex hull of the sensitivity space. 


In Section \ref{s:objective}, we extend the objective mechanism to allow for arbitrary $K$-mechs. We apply objective perturbation to the problem of logistic regression in Subsection \ref{s:logistic}, and derive the sensitivity space for this problem. In Subsection \ref{s:logisticSimulation}, we demonstrate through simulations that choosing the $K$-mech based on our criteria substantially improves the accuracy of the private logistic regression. In Section \ref{s:functional}, we modify the functional mechanism to allow for arbitrary $K$-mechs, with a focus on linear regression. For this problem the sensitivity space can be written in closed form, allowing us to implement the optimal $K$-mech exactly. We demonstrate the finite-sample utility gains of our approach in  Subsections \ref{s:linearSimulations} and \ref{s:housing} through simulations and a real data example.

We end with concluding remarks and discussion in Section \ref{s:conclusions}. For convenience, we collect algorithms to sample from several $K$-mechs in the appendix, Subsection \ref{s:sampling}. All proofs and some technical details are postponed to Subsection \ref{s:proofs}.

Throughout the paper, we aim to emphasize which results are our contribution and which are from the previous literature as follows: Any result/definition with a name and no citation is original to the present paper, and those with a citation are of course from the literature. Some minor definitions have no name or citation, but are common concepts in the fields of Statistics and Mathematics.

\section{Differential Privacy and Sensitivity Space}
\label{s:background}

In this section, we review the necessary background on differential privacy and propose the new concept that we refer to as the \emph{sensitivity space} of a statistic $T$, which extends the concepts of sensitivity polytope and sensitivity hull to the setting of nonlinear statistics $T$.

Differential privacy \citep{Dwork2006:Sensitivity}, provides a framework for a strong provable privacy protection against arbitrary adversaries while allowing the release of some statistics and potentially synthetic data. It requires the introduction of additional randomness into the analyses such that the distribution of the output does not change substantially if one person were to be in the database or not. A non-technical introduction to DP can be found in \citet{Nissim2017:Nontechnical}, and a comprehensive introduction can be found in \citet{Dwork2014:AFD}. 

Before we state the definition of DP, we recall the Hamming distance, which we use to measure the similarity of two databases. The definition of DP can easily be modified to allow for alternative metrics on the space of databases. 
\begin{defn}
Let $X,Y\in \mscr X^n$ for any space $\mscr X$. The \emph{Hamming distance} between $X$ and $Y$ is $\de(X,Y) = \#\{i\mid X_i \neq Y_i\}$, the number of entries where $X$ and $Y$ differ.
\end{defn}

\begin{defn}[\citealp{Dwork2006:Sensitivity}]\label{DP}
Let $X = (X_1,\ldots, X_n)\in \mscr X^n$. For $\ep>0$, a mechanism (random function) $M:\mscr X^n \rightarrow \mscr Y$ satisfies $\ep$-Differential Privacy ($\ep$-DP) if for all measurable sets $B$, and all pairs of databases $X$ and $X'$ such that $\de(X,X')=1$, we have that
 \[P(M(X') \in B\mid X')\leq \exp(\ep)P(M(X)\in B\mid X).\]
\end{defn}
Let's take a moment to discuss the cast of characters in Definition \ref{DP}. The value $\ep$ is called either the privacy parameter, or the privacy-loss budget. Smaller $\ep$ corresponds to more privacy, but as $\ep$ approaches infinity there is no privacy guarantee. We think of $X_i$ as the information provided by individual $i$, so  $\mscr X$ is the set of possible observations from one individual. We make no assumptions about the nature of  $\mscr X$: while it is common for $\mscr X$ to be a subset of $\RR^m$, $\mscr X$ could also be a set of networks, survey responses, or any other data structure. We call $\mscr Y$ the output space, which contains the possible values our statistic of interest can take on. In Definition \ref{DP} we do not place any restrictions on $\mscr Y$, however in this paper we  focus on $\mscr Y = \RR^m$ and use Lebesgue measure denote as $\lambda(\cdot)$. Finally the mechanism $M$ is our method of introducing randomness into the output, which is what achieves privacy.


A statistically insightful interpretation of Definition \ref{DP} was provided by \citet{Wasserman2010:StatisticalFDP}, connecting it to hypothesis testing. If an adversary wants to test whose data is in the $i^{th}$ entry of $X$ at level $\al$ test, Proposition \ref{DPHypothesis} assures us that the probability of rejecting the null hypothesis is small.

\begin{prop}[\citealp{Wasserman2010:StatisticalFDP}]
\label{DPHypothesis}
Suppose that $M:\mscr X^n\rightarrow \mscr Y$ satisfies $\ep$-DP, $P$ is a probability measure on $\mscr X^n$, and $Z=M(X)$ is the released output of the mechanism $M$. Then any level $\al$ test which is a function of $Z$, $M$, and $P$, of $H_0:X_i=u$ versus $H_1:X_i=v$ has power bounded above by $\al \exp(\ep)$.
\end{prop}

\begin{remark}
\label{DifferentDP}
Besides Definition \ref{DP}, 
other formulations of DP have been proposed; see \citet{Kifer2012axiomatic} for a formal axiomatization of privacy. For instance, by replacing $\mscr X^n$ with any set, and  $\de(\cdot,\cdot)$ with any metric on that set, one obtains a new notion of DP. Let $\mscr X^* = \{()\}\cup\mscr X\cup \mscr X^2\cup \cdots$ be the set of all finite tuples with entries in $\mscr X$, where $()$ represents the empty tuple. A popular alternative to Definition \ref{DP} takes $X\in \mscr X^*$, and $\de(X,X') = 1$ if $X$ can be obtained from $X'$ by adding or deleting an entry. We will refer to this notion as add/delete-DP, which appears in \citet{Dwork2014:AFD}. Note that $\ep$-add/delete-DP implies $2\ep$-DP. Other differences are that $n$ is not publicly known, and the mechanism must be well defined on any $X\in \mscr X^*$. While we use Definition \ref{DP} for concreteness, our theoretical results can be modified to accommodate these other settings.
\end{remark}

In this paper, we  study mechanisms that add a random vector to a statistic vector $T$. For such mechanisms, the variance of the random vector, and thus the additonal noise added, must be scaled differently depending on  the {\em sensitivity} of $T$, which captures the magnitude by which a single individual's data can change the output. The $\ell_1$-sensitivity was first introduced in \citet{Dwork2006:Sensitivity}, but sensitivity can be measured by other norms; e.g.,  $\ell_2$-sensitivity as in \citet{Chaudhuri2011:DPERM}. 

As we demonstrate in Sections \ref{s:logisticSimulation}, \ref{s:linearSimulations}, and \ref{s:housing}, the choice of norm for sensitivity calculations can have a significant impact on the performance of DP methods. Next we introduce a key new notion, the \emph{sensitivity space}, whose structure allows us to choose the best norm (See Section \ref{s:compare} for how to evaluate ``best'') for an arbitrary statistic in $\RR^m$. Notions similar to the sensitivity space have played an important role in the theoretical DP literature, but were limited to the study of linear statistics and often considered only $\ell_1$ norms on databases. In the case of linear statistics on binary-valued databases, the convex hull of the sensitivity space results in a linear transformation of the $L_1$ ball, which is referred to as the \emph{sensitivity polytope} \citep{Dwork2014,Dwork2015,Nikolov2015,Kattis2016}. \citet{Xiao2015,Xiao2017} study the setting of discrete location data, where the convex hull of the sensitivity space forms a two-dimensional polytope they call the \emph{sensitivity hull}. In this paper, the proposed sensitivity space is defined for {\em arbitrary} (not necessary linear) statistics in $\RR^m$, allowing for a wider variety of statistical applications, and generalizes the earlier notions of sensitivity polytope/hull. In Subsection \ref{s:explore}, we explore a simple example which demonstrates the complex geometry of the sensitivity space when applied to non-linear statistics. In Section \ref{s:compare}, we show that the optimal $K$-mech is determined by the convex hull of the sensitivity space. In fact, this is a fundamental result to know\footnote{We thank a reviewer for emphasizing this point.} in order to improve design of differentially private mechanisms in a formal and principled manner. In Sections \ref{s:logistic} and \ref{s:linear} we demonstrate this by using our more general notion of sensitivity space to optimize DP algorithms for linear and logistic regression.

\begin{defn}[Sensitivity Space]\label{AdjacentOutput}
Let $T:\mscr X^n\rightarrow \RR^m$ be any function. The \emph{sensitivity space} of $T$ is 
\[S_T = \l\{ u \in \RR^m \middle| \begin{array}{c}\exists X,X'\in \mscr X^n\text{ s.t. } \de(X,X')=1\\ \text{and }u=T(X)-T(X')\end{array}\r\}.\]
\end{defn}
The sensitivity space consists of all possible differences in the statistic vector $T$, when computed on two databases differing in one entry. 
Before  introducing the notion of sensitivity, we define \emph{norm balls}, which are the sets in one-to-one correspondence with norms.

\begin{defn}
A set $K\subset \RR^m$ is a \emph{norm ball} if $K$ is 
\begin{inparaenum}[1)]\item convex, \item bounded,  \item absorbing: $\forall u \in \RR^m$, $\exists c>0$ such that $u\in cK$, and \item
 symmetric about zero: if $u\in K$, then $-u\in K$. \end{inparaenum}
\label{NormBall}
\end{defn}

It is well known that if $K\subset \RR^m$ is a norm ball, then we can define a norm $\lVert \cdot \rVert_K: \RR^m \rightarrow \RR^{\geq0}$, given by $\lVert u \rVert_K = \inf \{c\in \RR^{\geq 0}\mid u\in cK\}$. We call $\lVert \cdot \rVert_K$ the \emph{$K$-norm}. In fact, any norm $\lVert \cdot \rVert$ can be generated this way by taking $K = \{u \mid \lVert u\rVert \leq 1\}$.

The sensitivity of a statistic $T$ is the largest amount that $T$ changes when one entry of $T$ is changed. Geometrically, the sensitivity of $T$ is the largest radius of $S_T$ measured by the norm of interest. 

\begin{defn}[$K$-norm Sensitivity]
For a norm ball $K \subset \RR^m$, the \emph{$K$-norm sensitivity} of $T$ is 
\[\Delta_K(T) = \sup_{\de(X,X')=1} \lVert T(X) - T(X')\rVert_K = \sup_{u \in S_T} \lVert u \rVert_K .\]

For $p\in[1,\infty]$, the \emph{$\ell_p$-sensitivity} of $T$ is $\Delta_p(T) = \sup_{u\in S_T} \lVert u\rVert_p$.
If $T$ is one-dimensional, we simply say the \emph{sensitivity} of $T$ is $\Delta(T) = \sup_{u \in S_T} |u|$.
\label{Sensitivity}
\end{defn}

The family of mechanisms we study in this paper are the $K$-Norm mechanisms ($K$-mechs), introduced in \citet{Hardt2010:GeometryDP}. While the focus of \citet{Hardt2010:GeometryDP} is on linear statistics, and uses a different metric on the input database, they remark that the mechanism is valid in more general settings as well. In Proposition \ref{prop:KNorm} and throughout the paper, we note Lebesgue measure as $\la(\cdot)$. Recall that the Lebesgue measure of a set $S$ can be interpreted as the volume of $S$. 

\begin{prop}[$K$-Norm Mechanism: \citealp{Hardt2010:GeometryDP}]\label{prop:KNorm}
%
Let $X  \in \mscr X^n$ and $T:\mscr X^n \rightarrow \RR^m$. Let $\lVert \cdot\rVert_K$ be any norm on $\RR^m$ and let $K = \{x\mid \lVert x \rVert_K\leq 1\}$ be its unit ball. Call
$\De_K(T) = \sup_{u\in S_T} \lVert u\rVert_K$.
Let $V$ be a random variable in $\RR^m$, with density $f_V(v) = \frac{\exp(\frac{-\ep}{\De}\lVert v\rVert_K)}{\Gamma(m+1)\la\left(\frac{\Delta}{\ep} K\right)}$, where $\Delta_K(T)\leq \Delta<\infty$. Then releasing $T(X)+V$, satisfies $\ep$-DP.
\end{prop}

Since all norms are equivalent in $\RR^m$, requiring that $\Delta_K(T)<\infty$ is equivalent to requiring that $S_T$ is bounded.

Since norms are symmetric about zero, any $K$-mech has mean $T(X)$ and so the $K$-mechs are a class of unbiased mechanisms, with respect to the randomness introduced for privacy. We refer to the $K$-mech with norm $\ell_p$ as the $\ell_p$-mechanism ($\ell_p$-mech). Note that in the case where the norm is $\ell_1$, the mechanism results in adding independent $\mathrm{Laplace}(\Delta/\ep)$ noise to each entry of $T$. So, the $K$-norm mechanisms can be viewed as a generalization of the Laplace mechanism.

Finally, Proposition \ref{PostComp} states that postprocessing cannot increase privacy risk. Postprocessing is both important as a privacy guarantee and as a useful tool to construct complex DP mechanisms (for example, the functional mechanism \citet{Zhang2012:FunctionalMR}).

\begin{prop}[Postprocessing: \citealp{Dwork2014:AFD}]\label{PostComp}
Let $X \in \mscr X^n$, $M:\mscr X^n\rightarrow \mscr Y$ be a random function, and $f:\mscr Y\rightarrow \mscr Z$ be any function. 
If $M$ is $\ep$-DP, then $f\circ M$ is $\ep$-DP.
\end{prop}


\subsection{Exploring Sensitivity Space}
\label{s:explore}
In this subsection, we provide a simple example which illustrates the relation between the sensitivity with respect to different norms, and the sensitivity space. In particular, we demonstrate the complex geometry of the sensitivity space when applied to non-linear statistics.

The $K$-norm sensitivity of $T$ is often studied as an algebraic object, being the supremum over a set of values. However, we can instead consider how it is geometrically related to the sensitivity space of Definition \ref{AdjacentOutput}. Geometrically, $\Delta_K(T)$ is the radius of the smallest $\lVert \cdot \rVert_K$-ball containing $S_T$. We study the $\ell_p$-sensitivity for the following extended example.

Throughout this paper the $\ell_1$, $\ell_2$, and $\ell_\infty$-mechs will make frequent appearances. All three of these mechanisms have been seen in the literature and they all have efficient sampling algorithms, which we detail in Section \ref{s:sampling}. On the other hand, while we consider arbitrary $K$-mechs later in the paper, we acknowledge that in general $K$-mechs are much harder to implement and sample from. We give a method to sample from arbitrary $K$-mechs in Subsection \ref{s:sampling}, via rejection sampling.
\begin{example}
\label{SquaredExample}
Suppose our database is $X = (X_1,\ldots, X_n) \in [-1,1]^n$, and our statistic of interest is $T(X) =(\sum_{i=1}^n X_i, \sum_{i=1}^n 2X_i^2)$. For this example, the sensitivity space is 
\begin{align*}
S_T&= \l\{ (u_1,u_2)\in \RR^2\middle| \begin{array}{ccl}u_1&=&x_1-x_2\\u_2&=&2x_1^2-2x_2^2\end{array}, \text{ for some } x_1,x_2\in [-1,1]\r\}\\
&= \l\{ (u_1,u_2) \in [-2,2]^2\middle |\  |u_2| \leq \begin{cases} 2-2(1-u_1)^2&\text{if } u_1\geq 0\\ 2-2(u_1+1)^2&\text{if } u_1<0\end{cases}\r\}.
\end{align*}
A plot of $S_T$ as a subset of $\RR^2$ is shown in Figure \ref{AOS}. For this example, we can work out the $\ell_1$, $\ell_2$, and $\ell_\infty$ sensitivities of $T$ exactly:
\begin{equation}
\Delta_1(T)=3.125, \qquad \Delta_2(T) =1/4\sqrt{71+8\sqrt{2}}, \qquad \Delta_\infty(T)=2.
\label{ExactSensitivity}
\end{equation}
Of these three, only $\Delta_\infty(T)$ can be computed by inspection. To compute $\Delta_1(T)$ and $\Delta_2(T)$, we solve the calculus problem $\max_{u_1\in [-2,2]} \lVert\binom{u_1}{u_2}\rVert$, where $u_2 = 2-2(u_1-1)^2$ (by symmetry of $S_T$, this is sufficient). The left plot of Figure \ref{AOS} shows the norm balls $\{ u\mid \lVert u\rVert_p\leq \Delta_p(T)\}$ for $p=1,2,\infty$.

\begin{figure}
\begin{center}
\includegraphics[width = .44\linewidth]{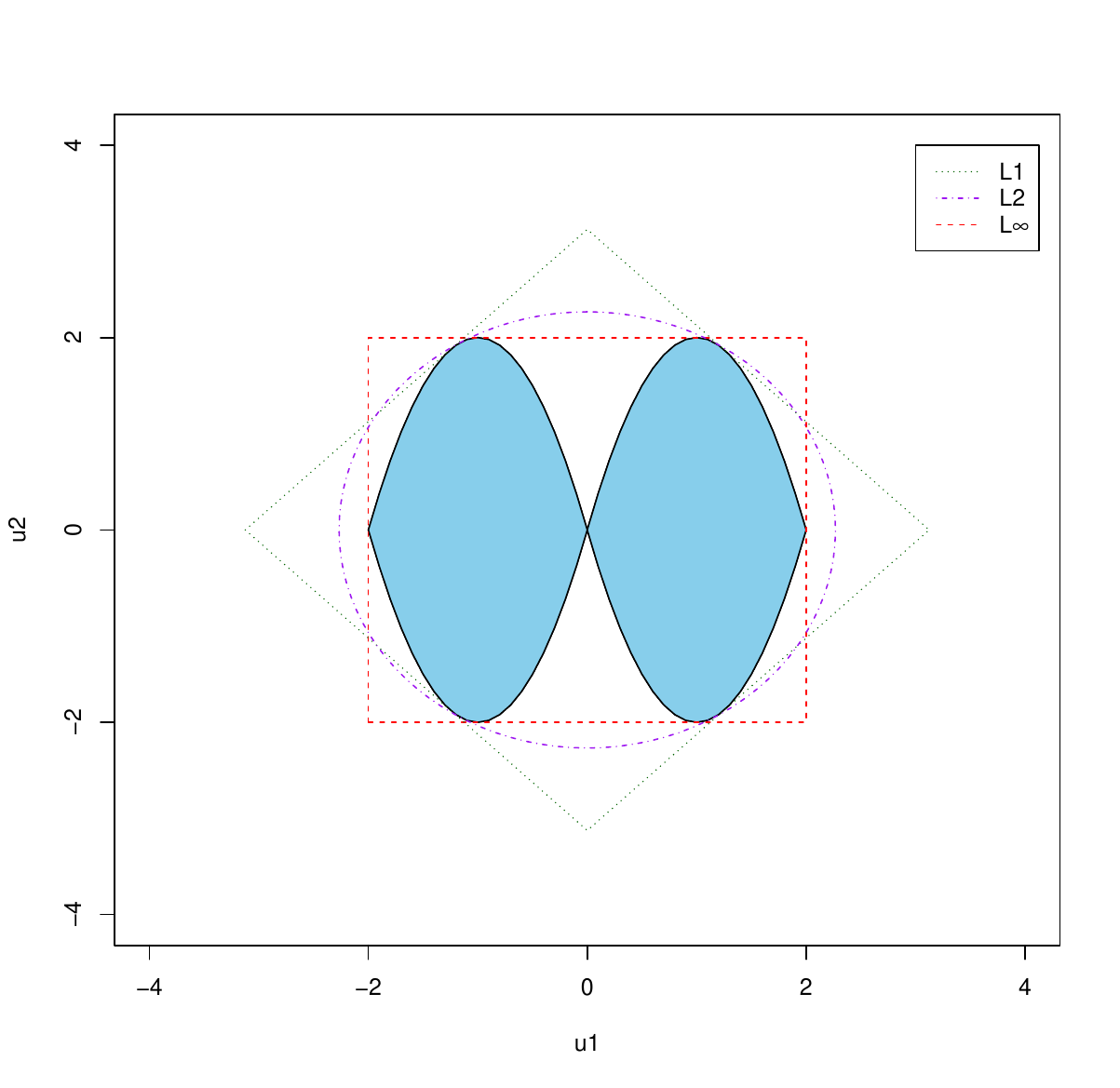}
\includegraphics[width = .44\linewidth]{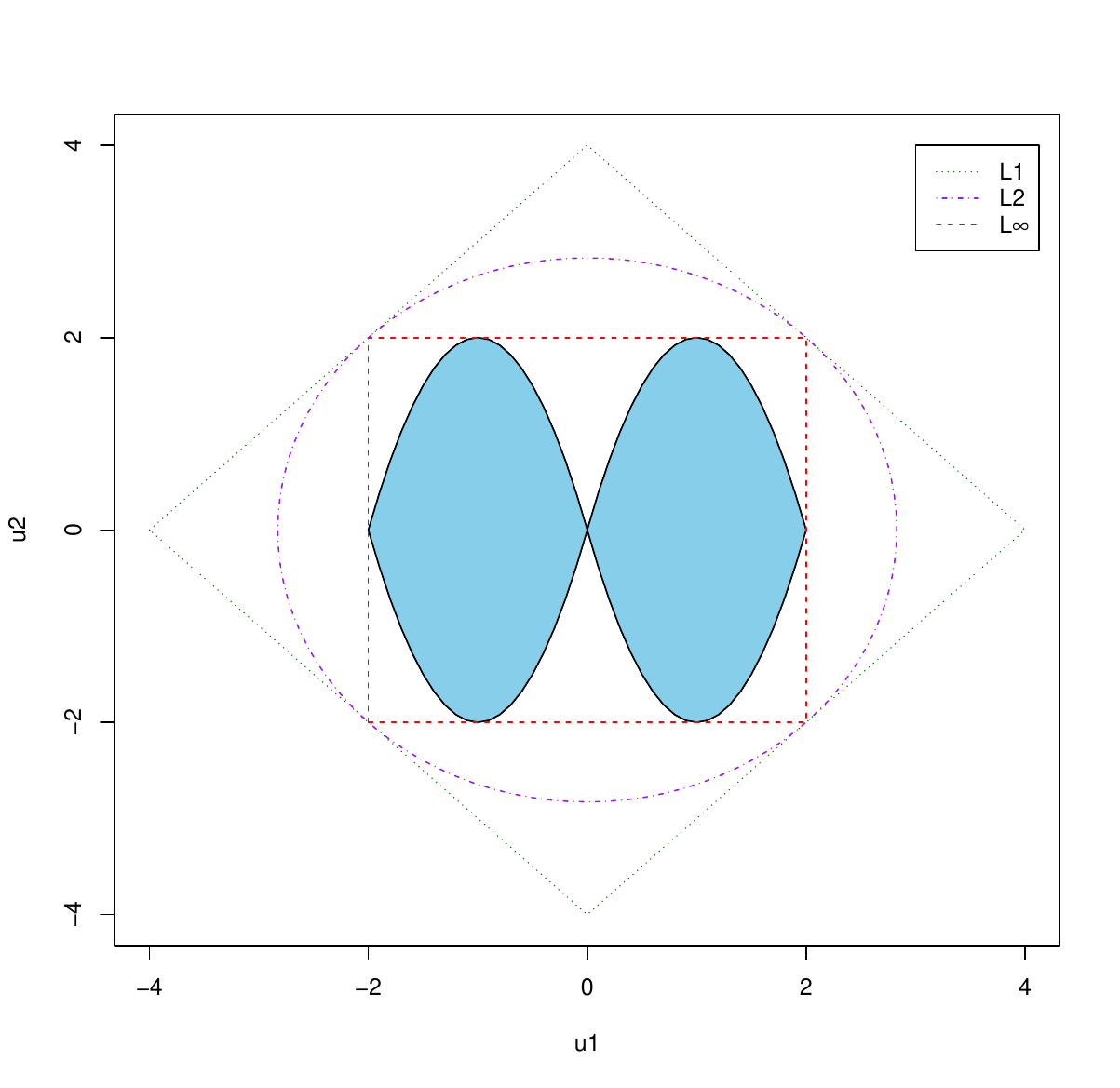}
\end{center}
\caption{The blue shaded area of both plots is the sensitivity space $S_T$ for Example \ref{SquaredExample}. In the left plot, the $\ell_p$ norm balls of radius $\Delta_p(T)$ as computed in Equation \eqref{ExactSensitivity}, for $p=1,2,\infty$ are plotted. In the right plot, the $\ell_p$ norm balls of radius $\Delta_p(T)$ using the approximations computed in Equation \eqref{ApproximateSensitivity}, for $p=1,2,\infty$ are plotted.}
\label{AOS}
\end{figure}

While for this example we are able to compute $\De_p(T)$ exactly for $p=1,2,\infty$, often in the literature, $\Delta_1(T)$ and $\Delta_2(T)$ are approximated in Equation \ref{ApproximateSensitivity}, visualized in the right plot of Figure \ref{AOS}. Such approximations are common; for example, see \citet{Zhang2012:FunctionalMR} and \citet{Yu2014}.  In these cases, the $\ell_1$ norm is used without considering other options, but by inspecting the right plot of Figure \ref{AOS} we see that $\ell_\infty$ is likely a better choice. After developing a formal and systematic method of choosing the best $K$-mech, we return to this problem in Examples \ref{ex:approx} and \ref{ex:volume} to confirm this intuition.

\begin{equation}
\begin{split}
\Delta_1(T)&= \sup_{\de(X,X')=1} \lVert T(X) - T(X')\rVert_1\leq \sum_{i=1}^2\sup_{\de(X,X')=1} |T_i(X) - T_i(X')|=4\\
\Delta_2(T)&
= \sup_{\de(X,X')=1} \sqrt{ \sum_{i=1}^2(T_i(X) - T_i(X'))^2}
\leq \sqrt{\sum_{i=1}^2\sup (T_i(X) - T_i(X'))^2 }= \sqrt 8.
\end{split}
\label{ApproximateSensitivity}
\end{equation}

 Note that when using these approximations, the norm balls $\{u\mid \lVert u\rVert_p\leq \Delta_p(T)\}$ are now larger, as seen in the right plot of Figure \ref{AOS}. On the other hand, in the left plot of Figure \ref{AOS}, we see that when we use the exact sensitivities via Equation \eqref{ExactSensitivity}, none of the norm balls we considered are contained in any of the others. In Section \ref{s:compare}, we develop criteria to determine the best norm ball in either scenario. If one norm ball is contained in another, we show in Theorems \ref{thm:tightness} and \ref{thm:variance} that in a strong sense the smaller norm ball is preferred. If however, neither norm ball is contained in the other, we propose using the norm ball with the smaller volume justified by Theorems \ref{thm:entropy} and Corollary \ref{cor:depth}. 
 In Subsections \ref{s:logisticSimulation}, \ref{s:linearSimulations}, and \ref{s:housing} we show that choosing the  $K$-mech based on our proposed criteria can have a substantial impact on finite sample utility. 
\end{example}

\section{Comparing $K$-Norm Mechanisms}
\label{s:compare}
In this section, we develop our main theoretical contributions that enable the comparison of $K$-norm mechanisms for a given statistic $T$ and sample size $n$, to determine which mechanism optimizes the finite sample utility. There are several different methods one could use to assess the optimality of a mechanism, for instance in terms of minimizing the expected value of a loss function. However in Section \ref{s:conclusions}, we provide a cautionary example, showing that for such objectives, the optimal mechanism changes for each loss function, and depends on the scaling of the statistic. Instead, we prefer to compare the mechanisms in more intrinsic measures which do not depend on the scaling or coordinate system used. In particular, we propose a novel framework consisting of three methods for comparing the $K$-mechs in terms of 1) stochastic ordering and statistical depth, 2) entropy, and 3) conditional variance. We show that each of these perspectives results in one of two stochastic orderings on the $K$-mechs. The first of the two orderings, which we call the \emph{containment order}, compares two $K$-mechs based on whether one associated norm ball is contained in the other. The other, called the \emph{volume order} is based on the volume of the associated norm balls. The containment order is a partial order, whereas the volume order is a total order which extends the containment order. We show that in both orderings, there is a minimal element which we call the ``optimal $K$-mech,'' whose corresponding norm is the convex hull of the sensitivity space.

\begin{defn}[Containment and Volume Orders]\label{def:orders}
  Let $V$ and $W$ be two random variables on $\RR^m$ with densities $f_V(v)= c \exp(-\frac{\ep}{\Delta_K}\lVert v\rVert_K)$ and $f_W(w) = c\exp(-\frac{\ep}{\Delta_H}\lVert w \rVert_H)$. We say that $V$ is preferred over $H$ in the \emph{containment order} if $\Delta_K\cdot K \subset \Delta_H \cdot H$. We say that $V$ is preferred over $H$ in the \emph{volume order} if $\la(\Delta_K \cdot K) \leq \la (\Delta_H \cdot H)$. 
\end{defn}

We propose a stochastic ordering which orders random variables based on the containment of certain ``level sets,'' which results in the containment order. We show that this stochastic ordering can also be motivated by concepts in statistical depth; based on the statistical depth literature, we arrive at an alternative stochastic ordering equivalent to the volume order. 
Another approach to comparing the $K$-norm mechanisms is based on their entropy, which we show is equivalent to the volume order. 
Finally, we compare the $K$-norm mechanisms in terms of their conditional variance, given a unit direction. This comparison also results in the containment order. 

In Figure \ref{fig:diagram}, a diagram illustrates our proposed framework, that is the relation between each of our perspectives, and the stochastic orderings associated with them. We can view the items in the top row as theoretical approaches to comparing the $K$-mechs, and items in the bottom row as decision criteria used to implement the comparison. 

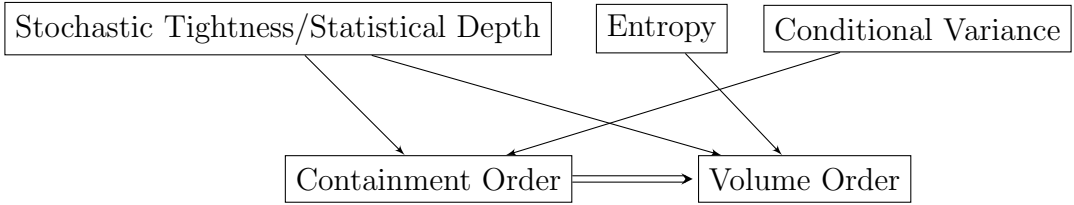
\begin{figure}
    \centering
\begin{tikzpicture}
\tikzset{vertex/.style = {draw,minimum size=1.5em}}
\tikzset{edge/.style = {->,> = latex'}}
\tikzstyle{vecArrow} = [thick, decoration={markings,mark=at position
   1 with {\arrow[semithick]{open triangle 60}}},
   double distance=1.4pt, shorten >= 5.5pt,
   preaction = {decorate},
   postaction = {draw,line width=1.4pt, white,shorten >= 4.5pt}]
\tikzstyle{innerWhite} = [semithick, white,line width=1.4pt, shorten >= 4.5pt]
\node[vertex] (stoch) at  (0,0) {Stochastic Tightness/Statistical Depth};
\node[vertex] (entropy) at (5.1,0) {Entropy};
\node[vertex] (cond) at (8.5,0) {Conditional Variance};

\node[vertex] (contain) at (2,-2) {Containment Order};
\node[vertex] (volume) at (7,-2) {Volume Order};

\draw[edge] (stoch) to (contain);
\draw[edge] (stoch) to (volume);
\draw[edge] (entropy) to (volume);
\draw[edge] (cond) to (contain);

 \draw[double,double distance=2pt,->,>=stealth,shorten >=2pt] (contain) to (volume);

\end{tikzpicture}
\caption{The top row represents the three theoretical perspectives, the bottom represents the two possible decision rules. A standard arrow indicates that a given theoretical perspective justifies the corresponding decision rule. The double arrow indicates that if two mechanisms are ordered with respect to containment, this implies they are ordered with respect to volume. }
\label{fig:diagram}
\end{figure}

\subsection{Stochastic Tightness}\label{s:stochastic}
In this section, we propose a multivariate extension of stochastic dominance, called \emph{stochastic tightness}\footnote{Not to be confused with the notion of \emph{tightness} from measure theory.} which orders the $K$-norm mechanisms based on the geometry of level sets we define and refer to as \emph{concentration sets}, and can be viewed as a multivariate extension of stochastic dominance. The idea is based on showing that the level sets containing probability $\alpha$ of one distribution are always contained in the corresponding level sets of another distribution. In Subsection \ref{s:depth}, we show that these level sets can also be motivated by concepts in the field of statistical depth. When applied to the $K$-norm mechanisms, we show that this ordering is equivalent to the containment order.

Recall that \emph{stochastic dominance} is a partial order on real-valued random variables, originally proposed in the context of decision theory \citep{Quirk1962}. A random variable $X$ stochastically dominates $Y$ if $F_X(t)\leq F_Y(t)$ for all $t\in \RR$, where $F(\cdot)$ represents the cumulative distribution function. 
Intuitively if $X$ stochastically dominates $Y$, then $X$ takes larger values than $Y$ in a strong sense.

There have been several efforts to extend stochastic dominance to 1) compare the dispersion of random variables about a center and 2) allow for the comparison of multivariate distributions. A detailed summary of these various extensions and their connection to statistical depth can be found in \citet{Zuo2000Nonparam}. 
Here, we introduce \emph{stochastic tightness}, a multivariate stochastic ordering with properties similar to stochastic dominance, which fits into the framework of \citet{Zuo2000Nonparam}. The idea is that the concentration sets, defined in Definition \ref{def:concentration} serve a similar purpose as the cumulative distribution function. While we use the notion of stochastic tightness in this paper to compare $K$-norm mechanisms, it is in fact applicable for a wider set of random variables and may be of independent interest in the statistical community.

First, we define the $\alpha$-concentration set of a random variable. 

\begin{defn}[Concentration Sets]\label{def:concentration}
  Let $X$ be a unimodal, continuous random variable on $\RR^m$ with center $a$. Assume further that $f_X$ is decreasing along any ray away from the center (i.e. $f(x)\leq f(\alpha x+(1-\alpha)a)$ for  all $x$ and all $\alpha \in [0,1]$), and that for all $t>0$, $P(\{x\mid f(x)=t\})=0$. For any $\alpha\in (0,1)$, the set $S_X^\alpha$ is defined to be the smallest (wrt Lebesgue measure) set such that $P(X\in S_{X}^\alpha)\geq \alpha$. We call $S_X^\alpha$ a \emph{$\alpha$-concentration set}.
\end{defn}

In Definition \ref{def:concentration}, the assumptions imply that the $\alpha$-concentration set is unique (see property 1) of Lemma \ref{lem:Sprops}). Note that all of these assumptions are easily verified for any $K$-mech.

\begin{defn}[Stochastic Tightness]\label{def:stoch}
  Let $X$ and $Y$ be random variables on $\RR^m$ with center $a$, which satisfy the assumptions of Definition \ref{def:concentration}. We say that $X$ is \emph{stochastically tighter about $a$} than $Y$ if for all $\alpha\in (0,1)$, $S_X^\alpha \subseteq  S_Y^\alpha$.
\end{defn}

Before we investigate the $\alpha$-concentration sets of the $K$-norm mechanisms, we require a lemma giving the marginal distribution of the magnitude of a $K$-norm random variable, sa measured by the $K$-norm. Conveniently, this works out as a gamma random variable. This result is similar to the decomposition given in \citet{Hardt2010:GeometryDP}, showing that a $K$-norm random variable can be generated by multiplying a gamma and a uniform random variable. In \citet{Kifer2012:PrivateCERM}, a $K$-norm random variable is referred to as a multivariate gamma distribution, which this result supports.

\begin{lem}[$K$-norm Marginal Distribution]\label{lem:gamma}
  For the random variable $V\in \RR^m$ with density $f_V(v) \propto \exp(-a \lVert v\rVert_K)$, the norm of $V$ is marginally distributed as  $\lVert V \rVert_K \sim \mathrm{Gamma}(m,a)$. 
\end{lem}

In Lemma \ref{lem:concentration}, we show that the $\alpha$-concentration sets for a $K$-norm random variable are dilations of $K$. This result is based on two facts: 1) for unimodal and continuous distributions, concentration sets are in 1-1 correspondence with level sets of the density \citep[Theorem 9.3.2]{Casella2002}, and 2) the density of a $K$-norm random variable is constant for values with the same $K$-norm. 

\begin{lem}[Concentration of $K$-mech]\label{lem:concentration}
  For the random variable $V\in \RR^m$ with density $f_V(v) \propto \exp(-a \lVert v\rVert_K)$, the $\alpha$-concentration set is
  \[S_V^\alpha = \{v\mid \lVert v\rVert_K \leq t\} = tK,\]
  where $t$ is the $\alpha$-quantile of $\mathrm{Gamma}(m,a)$.
\end{lem}

Finally, we combine Definition \ref{def:stoch} and Lemma \ref{lem:concentration} to show that one $K$-norm mechanism is stochastically tighter than another if its scaled norm ball is contained in the other.

\begin{thm}[Stochastic Tightness of $K$-mechs]\label{thm:tightness}
Let $K_V$ and $K_W$ be two norm balls in $\RR^m$, let $\Delta_V$ and $\Delta_W$ be positive real numbers. Define the two random variables $V$ and $W$ on $\RR^m$ with densities $f_V(v) \propto \exp\left(\frac{-\ep}{\Delta_V}\lVert v\rVert_{K_V}\right)$ and $f_W(w) \propto \exp\left(\frac{-\ep}{\Delta_W} \lVert w \rVert_{K_W}\right)$.  The random variable $V$ is stochastically tighter about zero than $W$ if and only if $\Delta_V \cdot K_V\subset \Delta_W \cdot K_W$.
\end{thm}

We see from Theorem \ref{thm:tightness} that the ordering based on stochastic tightness is equivalent to the containment ordering, as defined in Definition \ref{def:orders}. This gives us our first way of comparing the $K$-norm mechanisms to determine the optimal mechanism.

\begin{remark}
  While in the one-dimensional setting, stochastic dominance implies that the expected value of any increasing objective is maximized, an analogous property does not hold in multivariate settings. In fact, it may be the case that the $K$-mech is stochastically tighter than the $H$-mech, and yet many common loss functions are not optimized by $K$. See a cautionary example in Section \ref{s:conclusions}.
\end{remark}

\subsubsection{Statistical Depth}\label{s:depth}
Next, we construct a depth function whose depth regions coincide with the $\alpha$-concentration sets, proposed in the previous section. The $\alpha$-concentration sets are in fact a special case of \emph{depth regions}, sets determined by a \emph{depth function}. Based on this connection, the statistical depth literature provides additional justification for the stochastic tightness ordering, and also proposes a second ordering which we show is equivalent to the volume ordering.


First, we review the axioms and terminology of statistical depth, and then construct a depth function whose depth regions agree with the concentration sets of Definition \ref{def:concentration}. The first measure of statistical depth was introduced by \citet{Tukey1975}, which is commonly called the  half-space depth, or Tukey depth. An alternative to the Tukey depth was introduced by \citet{Liu1988,Liu1990}, and these works provide unifying axioms that characterize statistical depth. Such axioms were further formalized in \citet{Zuo2000General}. The concepts in statistical depth have also been connected to multivariate notions of order statistics, quantiles, and outlyingnesss measures \citep{Serfling2006}. See \citet{Mosler2013} for a review of the statistical depth literature and several examples of statistical depth functions.

\begin{defn}
  [Depth Function \citep{Zuo2000General}]\label{def:depth}
  A \emph{depth function} is a map $D:\RR^m \rightarrow \RR$ such that for any random variable $X$ on $\RR^m$ (or a specific sub-class of random variables), the following properties hold:
  \begin{itemize}
  \item [(A1)](Affine invariance) For any $x\in \RR^m$, any full rank $A\in \RR^{m\times m}$, and any $b\in \RR^m$, we have that $D_{{AX + b}}(Ax+b) = D_{X}(x)$.
  \item [(A2)] (Maximality at the center) If $x_0$ is the center of $X$, then $D_X(x_0)  =  \max_{x\in \RR^m} D_X(x)$
  \item [(A3)] (Linear monotonicity relative to the center) $D_X(x) \leq D_X((1-\alpha)x_0 + \alpha x)$ for all $\alpha \in [0,1]$ and all $x \in \RR^m$.
  \item [(A4)] (Vanishing at infinity) $\lim_{\lVert x \rVert \rightarrow \infty} D_X(x) = 0$.
  \end{itemize}
\end{defn}

Intuitively, a \emph{depth function} provides a measure of how close a point in $\RR^m$ is to the ``center'' of a given distribution, with higher values indicating that the point is closer. The first axiom requires that the depth function is invariant under affine transformations, so that the depth measure does not depend on the coordinate system. The third axiom says that the depth function decreases when moving away from the center.

One can use the depth function to draw contours of the distribution illustrating its shape. These contours are called \emph{depth regions}, sets of  points with depth greater than a given value. One can also construct depth regions with a specified probability content $\alpha$ as in Definition \ref{def:region}, which are conceptually similar to the $\alpha$-concentration sets of Definition \ref{def:concentration}. Depth regions with probability content $\alpha$ have also been referred to as \emph{depth lifts} \citep{Mosler2013}.

\begin{defn}
  [Depth Region \citep{zuo2000structural,Mosler2013}]\label{def:region}
  Given a depth function $D:\RR^m\rightarrow \RR$, the region of depth $d$ for the random variable $X$ is $C_X(d) \defeq \{x\in \RR^m \mid D_X(x) \geq d\}$. The \emph{depth region with probability content $\alpha$},  is $C_X(d(\alpha))$, where $d(\alpha) \defeq \inf\{d\in \RR\mid P_X(C_X(d))\geq \alpha\}$.
\end{defn}

  Based off of statistical depth, there are two natural stochastic orderings that have been considered in the literature. The first orders distributions based on the containment of depth regions \citep{Mosler2013}. The second order, proposed by \citet{Zuo2000Nonparam} orders two distributions based on the volume of the depth regions.

\begin{defn}
 [More Dispersed \citep{Mosler2013} and More Scattered  \citep{Zuo2000Nonparam}]
 Let $D:\RR^m \rightarrow \RR$ be a depth function , and let $X$ and $Y$ be two random variables on $\RR^m$. We say that $X$ is \emph{more dispersed} than $Y$ if $C_X(d(\alpha))\subset C_Y(d(\alpha))$ for all $\alpha\in (0,1)$. We say that $X$ is \emph{more scattered} than $Y$ if $\la(C_X(d(\alpha)))\geq \la(C_Y(d(\alpha)))$ for all $\alpha\in (0,1)$. 
\end{defn}


In Theorem \ref{thm:depth}, we construct a depth function whose depth regions agree with the $\alpha$-concentration sets. Verifying the axioms of Definition \ref{def:depth} requires a few technical properties of concentration sets, stated in Lemma \ref{lem:Sprops}, found in the Appendix. 

\begin{thm}[Depth for Concentration Sets]\label{thm:depth}
  Assume that a random variable $X$ on $\RR^m$ satisfies all of the conditions in Definition \ref{def:concentration}. Assume further that for all $x\in \RR^m$, $\lim_{t\rightarrow \infty} f_X(tx)=0$. Then $S_X^\alpha$ is a depth region with probability content $\alpha$ corresponding to the depth function $D_X(x) = 1- \inf\{\alpha\mid x\in S_X^\alpha\}$.
\end{thm}

Next we apply the concepts of \emph{more dispersed} and \emph{more scattered} to the $K$-norm mechanisms with our constructed depth function of Theorem \ref{thm:depth}, and show that these orderings coincide with the containment and volume orders, respectively.

\begin{cor}[$K$-mech More Dispersed/More Scattered]\label{cor:depth}
Let $V$ and $W$ be  $K$ and $H$-norm mechanisms respectively, on $\RR^m$. Based on the depth function in Theorem \ref{thm:depth}, $W$ is more dispersed than $V$ if and only if $\Delta_K \cdot K \subset \Delta_H\cdot H$, and
$W$ is more scattered than $V$ if and only if $\la(\Delta_K\cdot K) \leq \la(\Delta_H\cdot H)$. 
\end{cor}

Corollary \ref{cor:depth} provides us with a separate perspective to motivate the containment ordering and offers insight into the notion of stochastic tightness. We also see the volume order for the first time as related to \emph{more scattered}. Thus, both the stochastic tightness and statistical depth perspectives developed here for the $K$-norm mechanisms can lead to either the containment or volume order decision criteria. In the next subsection, we show that the volume ordering can also be motivated based on the entropy of the $K$-mechs.

\subsection{Entropy}\label{s:entropy}
In this section, we compute the entropy of the $K$-norm mechanisms and show that ordering the $K$-mechs based on entropy is equivalent to the volume order. 

The \emph{entropy} of a random variable is a concept introduced in the field of information theory, and was originally developed to communicate the amount of information that can be sent through a channel, or random variable \citep{Shannon1948}. For a general introduction to information theory, see \citep{Cover2012}. Roughly speaking, the greater the variability in a random variable, the greater the entropy. By minimizing the entropy of the noise adding distribution, we minimize the amount that the noise can corrupt the non-private signal.

There have been several works connecting the concepts of information theory and differential privacy. Some have studied alternative definitions of DP, phrased in terms of mutual information \citep{CUff2016,Wang2016}. A few notable works have derived minimum-entropy mechanisms under DP \citep{Wang2014,Wang2017}. These works show that for a database of one person, the minimum-entropy mechanism (with a differentiable density) is Laplace. They also show that for $d$-dimensional databases with sensitivity based on the $\ell_1$ norm, the minimum-entropy mechanism is iid Laplace. \citet{Duchi2013FOCS} derive mutual information bounds in the setting of local DP. \citet{Rogers2016} show that an information theoretic concept \emph{max-information} can be used to optimize DP mechanisms for the purpose of hypothesis testing.


\begin{defn}[Entropy \citep{Shannon1948}]\label{def:entropy}
  Let $X$ be a continuous random variable on $\RR^m$ with density $f_X(x)$. The \emph{entropy} of $X$ is $H(X)=\EE_{f_X}\left(-\log(f_X(X))\right)$. The \emph{conditional entropy} of $X$ given another random variable $Y$ is $H(X\mid Y) = \EE_{f_{X,Y}}\left(-\log(f_{X\mid Y}(X))\right)$, where the expectation is with respect to the joint distribution of $X$ and $Y$. The \emph{mutual information} of $X$ and $Y$ is $I(X,Y) = \EE_{f_{X,Y}} \log\left(\frac{f_{X,Y}(X,Y)}{f_X(X)f_Y(Y)}\right)$. 
\end{defn}
A useful identity is $I(X,Y) = H(X) - H(X\mid Y) = H(Y) - H(Y\mid X)$.

Next, we derive a closed-form formula for the entropy of a $K$-norm mechanism in Proposition \ref{prop:entropy}, and show in Theorem \ref{thm:entropy} that ordering $K$-mechs based on entropy is equivalent to the volume order.
\begin{prop}[Entropy of $K$-mech]\label{prop:entropy}
  Let $V$ be a random variable on $\RR^m$ with density $f_V(v) = \left(\ep/\Delta\right)^m \frac{\exp(-\frac{\ep}{\De} \lVert V \rVert_K)}{m! \lambda(K)}$. The entropy of $V$ is
  \[H(V) = \log\left( \left(\frac{\Delta \cdot e}{\ep} \right)^m m! \la(K)\right).\]
\end{prop}

From Proposition \ref{prop:entropy}, we see that the entropy of a $K$-norm mechanism is a linear function of $\log\left(\Delta^m \la(K)\right)$. So, minimizing the entropy of the mechanism is equivalent to minimizing the volume of $\Delta \cdot K$. Along with Corollary \ref{cor:depth}, we now have a second method of justifying the volume order. 

\begin{thm}[Entropy Ordering] \label{thm:entropy}
Let $K_V$ and $K_W$ be two norm balls in $\RR^m$. Consider the random variables $V$ and $W$ on $\RR^m$ with densities $f_V(v) \propto \exp(-\frac{\ep}{\Delta_V} \lVert v \rVert_{K_V})$ and $f_W(w) \propto \exp(-\frac{\ep}{\Delta_{W}} \lVert w \rVert_{K_W})$. We have that $H(V)\leq H(W)$ if and only if $\la(\Delta_V K_V)\leq \la(\Delta_W K_W)$. 
\end{thm}

We end this subsection by providing a connection between the entropy of the mechanism and the mutual information between the original statistic and the noisy output. Call $T$ the non-private statistic, $V$ the noise from a $K$-norm mechanism, and $Z=T+V$ the private output. It is intuitive that we would like to maximize the mutual information between $T$ and $Z$. However, this quantity depends on the distribution of $T$, which we assume is unknown. From another perspective, we would like to minimize $I(Z,V)$ which implies that the noise $V$ is not dominating the signal from $T$. A linear combination of these two objectives can be written in terms of the entropy of $T$ and $V$:
\begin{align*}
  I(T,Z) - I(Z,V)&=[H(Z) - H(Z\mid T)] - [H(Z) - H(Z\mid V)]\\
  &= H(T) - H(V),
\end{align*}
where we use the fact that $H(Z\mid V) = H(T)$. We can view $H(T)$ as an unknown constant. We see that maximizing the objective $I(T,Z) - I(Z,V)$ is equivalent to minimizing $H(V)$.

\subsection{Conditional Variance}\label{s:variance}

A natural question is whether there exists a $K$-mech which minimizes the variance in every direction. In this section, we show that given a direction in $\RR^m$, the conditional variance of one $K$-mech is smaller than another precisely when the associated norm balls are contained. 

First, we derive the distribution of a $K$-mech, conditioned on it lying in a one-dimensional subspace. The distribution is based on the distribution of $\lVert V \rVert_K$, developed in Lemma \ref{lem:gamma}.

\begin{lem}[$K$-mech Conditional Distribution]\label{lem:conditional}
  Let $\lVert\cdot\rVert_K$ be any norm on $\RR^m$. Let $V$ be a random variable with density $f_V(v)\propto\exp(-a  \lVert v\rVert_K)$. Then
  \begin{enumerate}
      \item The random variables $\lVert V \rVert_K$ and $\frac{V}{\lVert V \rVert_K}$ are independent.
      \item For any vector $e\in \RR^m$ with $\lVert e \rVert_2=1$, the distribution of $|V^\top e|$ conditional on $V\in \mathrm{span}(e)$, is $\mathrm{Gamma}(m,a\lVert e \rVert_K^{-1})$.
  \end{enumerate}
\end{lem}

The main result of this subsection, Theorem \ref{thm:variance} shows that a partial ordering of the $K$-mechs in terms of their conditional variance is in fact equivalent to the containment order. The proof of Theorem \ref{thm:variance} uses the conditional distribution developed in Lemma \ref{lem:conditional} and the observation that the variance of $\mathrm{Gamma}(m,\frac{\ep}{\Delta_K} \lVert e \rVert_K^{-1})$ is minimized by reducing the diameter of $\Delta_K\cdot K$ in the direction of $e$.

\begin{thm}[Conditional Variance of $K$-mechs]\label{thm:variance}
Let $K$ and $H$ be two norm balls in $\RR^m$. Let $\Delta_K$ and $\Delta_H$ be two positive real numbers. Consider the random variables $V_K,V_H \in \RR^m$ drawn from the densities $f(V_K) \propto \exp(\frac{-\ep}{\Delta_K}\lVert V_K\rVert_K)$ and $g(V_H) \propto \exp(\frac{-\ep}{\Delta_H} \lVert V_H\rVert_H)$.
If $\Delta_K \cdot K\subset \Delta_H\cdot H$, then for all $e\in \RR^m$ such that $\lVert e\rVert_2=1$,
\[\var\l(V_K^\top e\mid V_K \in \mathrm{span}(e)\r) \leq \var\l(V_H^\top e\mid V_H\in \mathrm{span}(e)\r).\]
\end{thm}
Theorem \ref{thm:variance} states that $V_K$ has uniformly smaller variance than $V_H$, conditional on any direction.

\subsection{Optimal $K$-Norm Mechanism}\label{s:optimal}
In Sections \ref{s:stochastic}, \ref{s:entropy}, and \ref{s:variance}, we provided various theoretical perspectives which in turn motivate either the containment or the volume order which provide two decision rules for determining the optimal $K$-mech. In this section we show that under mild assumptions, the containment and volume order both have the same minimal element, which is determined by the convex hull of the sensitivity space, a fundamental result that allows for a more principled design and evaluation of DP mechanisms.

As noted earlier, under either the containment order or the volume order, we prefer smaller norm balls which contain the sensitivity space $S_T$. Note that the convex hull of $S_T$ can be expressed as the intersection of all convex sets which contain $S_T$. Since all norm balls are convex, it follows that the convex hull of $S_T$ is a subset of any norm ball which contains $S_T$. So, if the convex hull of $S_T$ is a valid norm ball, then it corresponds to the optimal $K$-norm mechanism under either the containment or volume order. We formalize this observation in Theorem \ref{thm:optimal}.

Other works have proposed using the convex hull of the sensitivity space for $K$-norm mechanisms, but have not formalized it as a fundamental result for the development of DP mechanisms. In \citet{Xiao2015}, the convex hull is proposed for use in the $K$-Norm mechanism in the setting of two-dimensional discrete statistics. In \citet{Hardt2010:GeometryDP}, the linear transformations of $L_1$ balls are the convex hulls of the sensitivity space.

The following lemma establishes when the convex hull leads to a valid norm ball. 

\begin{lem}[Hull is Norm Ball]
\label{HullLem}
Let $T:\mscr X^n\rightarrow\RR^m$. 
Provided that $S_T$ is bounded and $\mathrm{span}(S_T)=\RR^m$, then $K_T = \mathrm{Hull}(S_T)$ is a norm ball. So,  the norm $\lVert \cdot \rVert_{K_T}$ is well defined.
\end{lem}

If $S_T$ is not bounded, then for any norm $\lVert\cdot \rVert_K$, the sensitivity $\Delta_K(T)$ is infinite, so no $K$-norm mechanism can be used to achieve $\ep$-DP. If $\mathrm{span}(S_T)$ is a proper subset of $\RR^m$, then the entries of $T(X)$ are linearly dependent. So, we can reduce the dimension of $T$, and recover the removed entries by post-processing.

\begin{thm}[Optimal $K$-mech]
\label{thm:optimal}
Let $T:\mscr X^n \rightarrow \RR^m$ such that $S_T$ is bounded and $\mathrm{span}(S_T)=\RR^m$. Let  $\lVert \cdot \rVert_K$ be any norm on $\RR^m$,  and consider the random variable $V_K \in \RR^m$ drawn from the density $f(V_K) \propto \exp\left(\frac{-\ep}{\Delta_K} \lVert \cdot \rVert_K\right)$. Then $V_{K_T}$ is preferred over $V_K$ in both the containment and volume orders.
\end{thm}

Provided that the conditions of Theorem \ref{thm:optimal} hold, it follows that the convex hull of the sensitivity space gives the $K$-mech which is least scattered, least dispersed, has minimum entropy, and has minimum conditional variance for all unit directions. Altogether, these properties justify calling this the optimal $K$-norm mechanism.

\subsubsection{Example of Containment and Volume Order}
\label{s:volume}
Here, we return to the setting of Example \ref{SquaredExample} and compare various $K$-mechs with the containment and volume orderings for that problem. We also provide a formula to compute the volume of $\ell_p$ balls, which appeared in \citep{Wang2005}, to simplify the calculations needed to apply the volume order.

\begin{example}\label{ex:approx}
Consider the setting of Example \ref{SquaredExample} to determine which mechanism to use based on the containment order. First we note that the convex hull gives a valid norm ball, so this norm is optimal and is written explicitly as $K_2$ in Section \ref{s:linear}. 

Between $\ell_1$, $\ell_2$, and $\ell_\infty$ we have two cases to consider. When using the exact sensitivities in Equation \eqref{ExactSensitivity}, illustrated in the left plot of Figure \ref{AOS}, we see that no norm ball is contained in another. Thus, these $K$-mechs are incomparable with respect to the containment order, and we cannot determine which mechanism is preferred using this decision criteria. On the other hand, in Example \ref{ex:volume}, we show that the volume order is able to compare these mechanisms. If instead, the sensitivities are approximated as in Equation  \eqref{ApproximateSensitivity}, illustrated in the right plot of Figure \ref{AOS}, we see that the norm balls are strictly contained. Thus, the containment order prefers the mechanisms from best to worst as $\ell_\infty$, $\ell_2$ and finally $\ell_1$.
\end{example}

An issue with the containment order is that in higher dimensions, determining containment can be nontrivial. On the other hand, computing volume is relatively simple. In particular, for $\ell_p$-balls there is a convenient formula, provided in \citep{Wang2005} stated in Proposition \ref{prop:volume}. 

\begin{prop}[\citealp{Wang2005}]\label{prop:volume}
  The volume of a unit $\ell_p$ ball in $\RR^m$ is $\frac{2^m \Gamma(1+1/p)^m}{\Gamma(1+m/p)}$.
\end{prop}

\begin{example}\label{ex:volume}
  Returning to Example \ref{SquaredExample}, using the exact sensitivities of Equation \ref{ExactSensitivity} as illustrated in the left plot of Figure \ref{AOS}, the volumes for the $\ell_1$, $\ell_2$ and $\ell_\infty$ balls are $\approx 19.53$, $\approx 16.16$, and $16$, respectively. So, based on the volume order of these three $K$-norm mechanisms, we prefer the $\ell_\infty$-mech in this setting. While we only considered $\ell_1$, $\ell_2$, and $\ell_\infty$ balls in Example \ref{SquaredExample}, by Theorem \ref{thm:optimal} we {now know} that the optimal $K$-mech is produced by using the convex hull of the sensitivity space. In fact, we are able to compute the volume of the convex hull of the sensitivity space as $\approx 13.33$, which offers an even better utility than the $\ell_\infty$-mechanism. 
\end{example}

\section{Generalization of Objective Perturbation}
\label{s:objective}
In this section we propose a generalization of the objective perturbation mechanism to allow for arbitrary $K$-norm mechanisms. In Subsection \ref{s:logistic}, we use the techniques of Section \ref{s:compare} to determine the best $K$-norm mechanisms for use in logistic regression. In Subsection \ref{s:logisticSimulation}, we demonstrate through simulations that the choice of mechanism can have a substantial impact on statistical utility. 

The objective perturbation mechanism was introduced in \citet{Chaudhuri2009:Logistic} for the application of logistic regression. In \citet{Chaudhuri2011:DPERM}, the mechanism was extended to general empirical risk problems, and further  extended in \citet{Kifer2012:PrivateCERM} and \citet{Yu2014}. 

\begin{algorithm}
\caption{Objective Perturbation as stated in \citet{Kifer2012:PrivateCERM}}
\scriptsize
INPUT: $X\in \mscr{ X}^n$, $\ep>0$, a convex set $\Theta \subset \RR^m$, a convex function $r: \Theta\rightarrow \RR$, a convex loss  $\hat{\mscr  L}(\ta; X) =\frac1n \sum_{i=1}^n \ell(\ta;x_i)$ defined on $\Theta$ such that the Hessian $\nabla^2 \ell(\ta;x)$ is continuous in  $\ta$ and $x$, $\De>0$ such that $\lVert \nabla \ell(\ta;x)\rVert_2\leq \De$ for all $\ta\in \Theta$ and $x\in \mscr X$, and $\la>0$  is an upper bound on the eigenvalues of $\nabla^2\ell(\ta;x)$ for all $\ta\in \Theta$ and $x\in \mscr X$.
\begin{algorithmic}[1]
  \setlength\itemsep{0em}
  \STATE Set $\ga = \frac{2\la}{\ep}$
\STATE Draw $V\in \RR^m$ from the density $f(V;\ep, \De)\propto \exp(-\frac\ep{2\De}\lVert V\rVert_2)$
\STATE Compute $\ta_{DP} = \arg\min_{\ta\in \Theta} \hat{\mscr L}(\ta;X) +\frac{1}nr(\theta)+ \frac{\ga}{2n} \ta^\top \ta + \frac{V^\top \ta}{n}$
\end{algorithmic}
OUTPUT: $\ta_{DP}$
\label{KiferAlgorithm}
\end{algorithm}

In \citet{Kifer2012:PrivateCERM}, it is shown that the output of Algorithm \ref{KiferAlgorithm} satisfies the add/delete formulation of DP, discussed in Remark \ref{DifferentDP}. We need to modify the algorithm to satisfy Definition \ref{DP}. Based on the proof in \citet{Kifer2012:PrivateCERM} we make several observations.
First, to have the output satisfy Definition \ref{DP}, we require that $\sup_{x,x'\in \mscr X} \sup_{\ta\in \Theta} \lVert \nabla \ell(\ta;x) - \nabla \ell(\ta;x')\rVert\leq \De$. This is related to our notion of sensitivity, but with the inclusion of the parameter $\theta$.  Next the use of $\ell_2$ norm to measure the sensitivity, and its use in the density of step 2, is arbitrary. \citet{Yu2014} note that $\ell_1$ can be used in place of $\ell_2$. In fact, any norm can be used along with its $K$-mech, and so the decision criteria of Section \ref{s:compare} can be applied. Furthermore, we can reduce the size of $\ga$ by taking $\ga = \frac{\la}{e^{\ep/2}-1}\leq \frac{2\la}{\ep}$ (see also \citet{Yu2014}). 
 Finally, to control the trade-off between bias and variance, we can introduce a tuning parameter $0<q<1$ and replace $f(V;\ep,\De) \propto \exp(-\frac{\ep q}{\De} \lVert V\rVert_K)$ and $\ga =\frac{\la}{e^{\ep(q-1)}-1}$. In Algorithm \ref{KiferAlgorithm}, $q$ is fixed at $1/2$.
 Incorporating these observations,  in Algorithm \ref{ExtendedObjPert}, we propose a {\em generalized objective perturbation mechanism}.

\begin{algorithm}
\caption{Extended Objective Perturbation}
\scriptsize
INPUT: $X\in \mscr{ X}^n$, $\ep>0$, a convex set $\Theta \subset \RR^m$, a convex function $r: \Theta\rightarrow \RR$, a convex loss  $\hat{\mscr  L}(\ta; X) =\frac1n \sum_{i=1}^n \ell(\ta;x_i)$ defined on $\Theta$ such that the Hessian $\nabla^2 \ell(\ta;x)$ is continuous in $\ta$ and $x$, $\De>0$ such that $\sup_{x,x'\in \mscr X} \sup_{\ta\in \Ta}\lVert \nabla \ell(\ta;x) - \nabla\ell(\ta;x')\rVert_K\leq \De$ for some norm $\lVert \cdot \rVert_K$,  $\la>0$  is an upper bound on the eigenvalues of $\nabla^2\ell(\ta;x)$ for all $\ta\in \Theta$ and $x\in \mscr X$, and a real value $0<q<1$.
\begin{algorithmic}[1]
  \setlength\itemsep{0em}
  \STATE Set $\ga = \frac{\la}{\exp({\ep(q-1)})-1}$
\STATE Draw $V\in \RR^m$ from the density $f(V;\ep, \De)\propto \exp(-\frac{\ep q}{\De}\lVert V\rVert_K)$
\STATE Compute $\ta_{DP} = \arg\min_{\ta\in \Theta} \hat{\mscr L}(\ta;X) +\frac1n r(\theta)+ \frac{\ga}{2n} \ta^\top \ta + \frac{V^\top \ta}{n}$
\end{algorithmic}
OUTPUT: $\ta_{DP}$
\label{ExtendedObjPert}
\end{algorithm}

\begin{thm}[Extended Objective Perturbation]
The output of Algorithm \ref{ExtendedObjPert} satisfies $\ep$-DP.
\label{ThmObjPert}
\end{thm}

The proof of Theorem \ref{ThmObjPert} mimics the proof in \citet{Kifer2012:PrivateCERM}, and can be found in the Appendix.

\subsection{Logistic Regression via Objective Perturbation}
\label{s:logistic}

In this subsection, we apply the  objective perturbation mechanism from Algorithm \ref{ExtendedObjPert} to the problem of logistic regression. We detail the sensitivity space for this problem, and compare the $\ell_1$, $\ell_2$ and $\ell_\infty$ mechanisms based on the containment and volume orderings of Section \ref{s:compare}.

Our setup is as follows: we observe $X_{ij}\in [-1,1]$ and $Y_i \in \{0,1\}$ for $i=1,\ldots, n$ and $j=1,\ldots, m$ \footnote{It may be necessary to rescale and truncate $X$ such as in \citet{Zhang2012:FunctionalMR} and \citet{Lei2016}.}. We take our loss function to be the negative log-likelihood of the logistic regression model:
\begin{equation}
\label{LossLogistic}
\hat {\mscr L} (\ta;X,Y) = \frac{1}{n} \sum_{i=1}^n \ell(\ta;X_i,Y_i) = \frac1n \sum_{i=1}^n \log(1+\exp(\ta^\top X_i))-Y_i \ta^\top X_i.
\end{equation}
The gradient and hessian of $\ell$ are 
\[
\nabla \ell (\ta;X_i,Y_i) = \l(\frac{\exp(\ta^\top X_i)}{1+\exp(\ta^\top X_i)}-Y_i\r) X_i,\qquad
\nabla^2 \ell(\ta;X_i,Y_i) = \l(\frac{\exp(\ta^\top X_i)}{(1+\exp(\ta^\top X_i))^2}\r) X_i X_i^\top
\]
By inspection, we note that the  eigenvalues of $\nabla^2\ell(\ta;x,y)$ are bounded above by $\la = \frac{m}{4}$. This bound is tight, by taking $X_i = (1,\ldots, 1)^\top$ and $\theta = (0,\ldots, 0)^\top$.

For  objective perturbation, the sensitivity space is slightly different than we defined in Definition \ref{AdjacentOutput}. Instead, we also allow for all values of $\theta$:
\begin{align*}
S &= \l\{ u \in \RR^m \middle| u = \nabla\ell(\ta;X_1,Y_1) - \nabla\ell(\ta;X_2,Y_2), \text{ s.t. } X_1,X_2\in [-1,1], Y_1,Y_2\in \{0,1\}, \text{ and }\ta\in \RR^m\r\}\\
&= \l\{ \l(\frac{\exp(\ta^\top X_1)}{1+\exp(\ta^\top X_1)} -Y_1\r)X_1 - \l( \frac{\exp(\ta^\top X_2)}{1+\exp(\ta^\top X_2)}-Y_2\r)X_2\r\}.
\end{align*}
Note that $\l( \frac{\exp(\ta^\top x)}{1+\exp(\ta^\top x)}\r)\in [0,1]$ no matter $\ta$ or $x$. We see that the per entry sensitivity of $\nabla \ell$ is bounded above by $2$. So, $\Delta_\infty(\nabla \ell) \leq 2$ and $S\subset [-2,2]^m$. For $m\geq 2$, we have found via simulations that the set 
$\l\{c\in \RR^m \middle| \exists k \text{ s.t. } c_i \in \{-2,2\} \text{ for } i\neq k \text{ and } c_k \in\{-1,1\}\r\}$.
is contained in $S$. This suggests that while the $\ell_\infty$-norm may not be optimal, as $m$ increases $\ell_\infty$ gets closer and closer to optimal. From this, we get the approximate sensitivities
$\Delta_\infty(\nabla \ell)  = 2$, $\Delta_2(\nabla \ell) = 2\sqrt{m}$,  and $\Delta_1(\nabla\ell) = 2m$. 
In fact these sensitivity calculations place us in a setting similar to the right plot of Figure \ref{AOS}, where the $\ell_\infty$, $\ell_2$, and $\ell_1$ balls are contained, in that order. Thus, by either the containment or the volume ordering, we have by Theorems \ref{thm:tightness}, \ref{thm:entropy}, and \ref{thm:variance} that $\ell_\infty$ is the preferred $K$-mech of these three options. Furthermore, as the convex hull of the sensitivity space is only slightly smaller than the $\ell_\infty$ ball, we have that the $\ell_\infty$-mech is nearly optimal in the sense of Theorem \ref{thm:optimal}.

\subsection{Logistic Regression Simulations}
\label{s:logisticSimulation}
In this section, we implement Algorithm \ref{ExtendedObjPert} for logistic regression on simulated data. We know from our analysis in the previous section along with the results of Section \ref{s:compare} that $\ell_\infty$-mech should outperform the $\ell_1$ or $\ell_2$ mechanisms. We show through simulations that the performance gains by choosing the $\ell_\infty$-mech are substantial, demonstrating that our choice of $K$-norm mechanism improves statistical utility. 

Our simulation procedure is described in Algorithm \ref{LogisticAlgorithm}. For a DP estimate $\beta_{DP}$, we measure its performance as the $\ell_2$ distance to the true $\beta$: $\ds\lVert \beta_{DP} - \beta\rVert_2$. We set $n=10^4$ and consider $\ep \in \{1/64,1/32,\ldots, 1,2\}$. The DP methods we implement are $\ell_1$-mech with $\De_1=2m$ and $q=1/2$, $\ell_2$-mech with $\De_2=2\sqrt m$ and $q=1/2$, and $\ell_\infty$-mech with $\De_\infty=2$ and $q\in \{1/2,.85\}$. First we compare $\ell_1$, $\ell_2$, and $\ell_2$ with $q=1/2$ as this is the value used in \citet{Chaudhuri2009:Logistic,Chaudhuri2011:DPERM,Kifer2012:PrivateCERM,Yu2014}. We chose $q=.85$ to show that performance can be further improved by tuning $q$. Unfortunately, we do not tune $q$ under DP so this limits its current usability. 

 \begin{algorithm}
 \caption{Logistic Regression on Simulated Data}
\scriptsize
 INPUT: $\ep$ and $n$
 \begin{algorithmic}[1]
  \setlength\itemsep{0em}
 \STATE Set $\beta = (0,-1,\frac{-1}{2},\frac{-1}{4},0,\frac{3}{4},\frac{3}{2})$ and $m=7$
\FOR{each of 100 replicates}
  \STATE Draw $X_{ij} \iid U[-1,1]$ and $U_{i} \iid U[0,1]$ for $i=1,\ldots, n$ and $j = 1,\ldots, m$
  \STATE Set $Y_{i} = \begin{cases} 1&\text{if } U_i < {e^{X\beta}}/(1+e^{X\beta})\\
0&\text{otherwise}\end{cases}$
\ENDFOR
\FOR{each DP estimate and each replicate $(X,Y)$}
\STATE Compute DP estimate $\beta_{DP}$ via Algorithm \ref{ExtendedObjPert}
\STATE Compute $\ell_2$ distance to $\beta$: $L_{DP} = \lVert \beta_{DP} - \beta\rVert_2$
\ENDFOR
\end{algorithmic}
OUTPUT: 
$\displaystyle\mathrm{median}_{\text{replicates}} \{L_{DP}\}$ for each method of DP.
 \label{LogisticAlgorithm}
 \end{algorithm}
\fig{width = .7\linewidth}{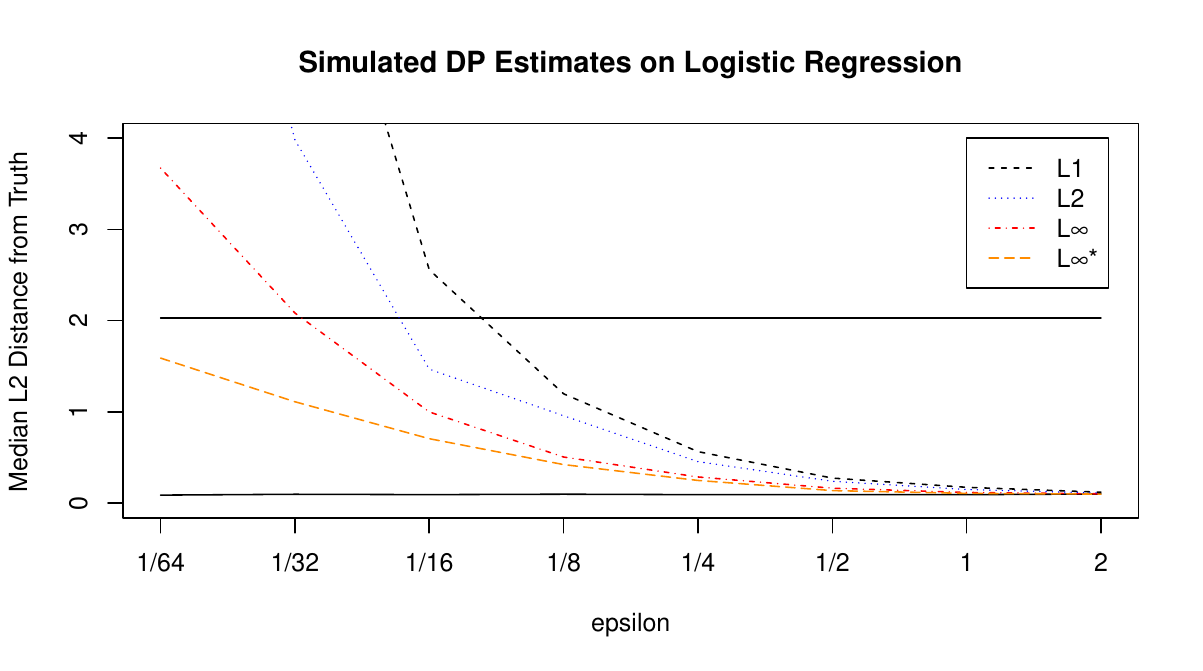}{Comparison of $\ell_1$, $\ell_2$, and $\ell_\infty$-mechanisms for logistic regression on simulated data, measured by $\ell_2$ distance to the true $\beta$. The estimates are via Algorithm \ref{ExtendedObjPert}, and the simulation procedure is described in Algorithm \ref{LogisticAlgorithm} with $n=10^5$. For all estimates, we use $q=1/2$ except for $L_\infty^*$, which uses $q=.85$. For each $\ep$, 100 replicates are used. The upper solid horizontal line indicates the distance between the zero vector and the true $\beta$. The lower solid line indicates the distance between the MLE $\hat \beta$ and the true $\beta$.
}
%
In Figure \ref{fig:L2LogisticExtra}, the $x$-axis indicates the value of $\ep$, and the $y$-axis is the median $\ell_2$ distance between the DP estimates and the true $\beta$. 
In this plot, we see that when we fix $q=1/2$, $\ell_\infty$ is better than $\ell_2$, which beats $\ell_1$; for example, $\ell_\infty$ saves approximately twice the privacy-loss budget $\epsilon$ compared to $\ell_2$ in this case, which is particularly important for small values of $\ep$. Specifically, the $\ell_\infty$ mechanism achieves a utility value of approximately 1 at $\ep=1/16$, whereas for $\ell_1$ to achieve a similar utility, it requires $\ep=1/8$. Recall that the $\ell_2$ is the norm used in \citet{Chaudhuri2009:Logistic}, \citet{Chaudhuri2011:DPERM} and \citet{Kifer2012:PrivateCERM}, and $\ell_1$ is used in \citet{Yu2014}. In \citet{Yu2014}, they argue that the $\ell_1$ should give better performance than $\ell_2$, which contradicts our result here. In their analysis, however they use $\Delta_1$ as an approximation for $\Delta_2$ which hinders the performance of $\ell_2$. Instead of either $\ell_1$ or $\ell_2$, we recommend the $\ell_\infty$ norm for this application, as the performance gains are substantial. Gains like this could have significant impact on real life applications and usability of DP mechanisms. 

We also include the $\ell_\infty$-mechanism with $q=.85$, labeled as $L_{\infty^*}$ in Figure~\ref{fig:L2LogisticExtra}. This tuning value $q$ was not chosen under DP, but does demonstrate that utility can be even further improved by considering other values of $q$. The choice of $q$ under DP is left to future researchers. 

\section{Linear Regression via Functional Mechanism }
\label{s:functional}
In this section, we show that the functional mechanism, a natural mechanism for linear regression, can be easily modified to allow for arbitrary $K$-norm mechanisms as well. We show that the convex hull of the sensitivity space can be written explicitly in this case, allowing for exact implementation of the optimal $K$-norm mechanism as determined by Theorem \ref{s:optimal}. In Subsection \ref{s:linearSimulations}, 
we demonstrate that the optimal $K$-mech improves the accuracy of the privatized estimates, as measured by the confidence interval coverage.  
In Subsection \ref{s:housing}, we show through a real data example that the choice of $K$-mech reduces the noise introduced in the privatized estimates compared to the non-private estimates. 
\subsection{Linear Regression Setup}
\label{s:linear}
Consider the setting where we have as input $X$, a $n\times(p+1)$ matrix with left column all $1$, and $Y$ a $n\times 1$ vector such that $X_{ij},Y_i\in [-1,1]$ for all $i,j$ \footnote{It may be necessary to rescale and truncate $X$ and $Y$ such as in \citet{Zhang2012:FunctionalMR} and \citet{Lei2016}.}. We want to estimate $\beta$ in the model $Y=X\beta+e$, where $e\sim N(0,\sas I)$. 
There are many ways of estimating $\beta$ under DP, such as those discussed in Section \ref{s:conclusions}.
Our approach in this section is to sanitize $X^\top X$ and $X^\top Y$ by either $\ell_1$-mech, $\ell_\infty$-mech, or the optimal $K$-mech, and obtain an estimate of $\beta$ via post-processing. This approach is similar to the functional mechanism \citep{Zhang2012:FunctionalMR}, which adds noise to the coefficients of the squared-loss function, before minimizing it. The differences in our approach compared to that in \citet{Zhang2012:FunctionalMR} are \begin{inparaenum}[1)]\item a tighter sensitivity analysis which requires less noise, \item an extension to any $K$-mech rather than just $\ell_1$-mech, and \item the use of Equation \eqref{DPEstimate} rather than minimizing the perturbed loss function (which results in more stable estimates)\end{inparaenum}. These three differences provide a useful extension of the work in \citet{Zhang2012:FunctionalMR}, resulting in better utility under $\ep$-DP.

Let $T$ be the vector of unique, non-constant entries of $X^\top X$ and $X^\top Y$ ($T$ has length $d = [\frac12 (p+1)(p+2)-1]+[p+1]$). From the sanitized version of $T$, we can recover approximations of $X^\top X$ and $X^\top Y$, which we call $(X^\top X)^*$ and $(X^\top Y)^*$ respectively. Then, using the postprocessing property of DP, Proposition \ref{PostComp},  our DP estimate of $\beta$ is 
\begin{equation}
\label{DPEstimate}
\hat\beta^* = [(X^\top X)^*]^\dagger(X^\top Y)^*,
\end{equation}
where $A^\dagger$ denotes the Moore-Penrose pseudoinverse of matrix $A$.

We sanitize $T$ by adding noise from either $\ell_1$, $\ell_\infty$, or the optimal $K$-mech. We rescale the elements of $T$ so that they all have sensitivity $2$. For example, since all $X_{ij}\in [-1,1]$ the sensitivity of $\sum_i X_{ij}$ is 2, but the sensitivity of $\sum_{i} X^2_{ij}$ is only 1. So, we replace $\sum_{i} X_{ij}^2$ with $2\sum_{i} X_{ij}^2$. The value in rescaling this way is demonstrated in a cautionary example in Section \ref{s:conclusions}. After adding noise, we can divide by $2$ to recover our estimate of $\sum_{i} X^2_{ij}$.

In order to implement the optimal $K$-mech, we use Algorithm \ref{SampleK} found in Section \ref{s:sampling} which requires us to sample uniformly from $K_T$, the convex hull of the sensitivity space 
\[S_T =\l \{u\in \RR^m\mid   
\begin{array}{c}
\exists\  \de\l((X,Y),(X',Y')\r)=1 \\
\text{ s.t. } u = T(X,Y)-T(X',Y')
\end{array}\r\}.\]
To understand the geometry of $K_T$,
we consider the two following subproblems:
\begin{itemize}
\item The convex hull of the sensitivity space for $(\sum_i X_i, \sum_i 2X_i^2)$, where $X_i\in [-1,1]$ is
\[K_2 =\l\{(u_1,u_2)\in[-2,2]^2 \text{ s.t. } 
|u_2|\leq
\begin{cases}
2-2(u_1-1)^2,&\text{if }u_1>1\\
2-2(u_1+1)^2,&\text{if }u_1<-1
\end{cases}
\r\}.\]
Note that $K_2$ is the sensitivity space studied in Example \ref{SquaredExample}. 
\item Suppose we want to release $(\sum_i X_i,\sum_i Y_i,\sum_i X_iY_i)$ where $X_i,Y_i\in [-1,1]$. The convex space for this statistic vector is
$K_{3} = \l\{ (u_1,u_2,u_3)\in [-2,2]^3 \text{ s.t. } |u_1|+|u_2|+|u_3|\leq 4\r\}.$
\end{itemize}
 For brevity, we omit the arguments that these are indeed the correct convex hulls. Then, $K_T$ is the $d$-dimensional convex set, which consists of several copies of $K_{2}$ and $K_{3}$ in different subspaces. This characterization of $K_T$ allows us to determine if a given vector lies in $K_T$. Using this, we are able to sample from the optimal $K$-mech via Algorithm \ref{RejectionAlgorithm} in Section \ref{s:appendix}. 
\subsection{Linear Regression Simulations}\label{s:linearSimulations}
In this section, we measure how close the estimates generated by \eqref{DPEstimate} are to the true $\beta$, for each DP mechanism. We consider an estimate close enough to the true $\beta$ if each entry of the estimate is in the $95\%$ non-private confidence interval (CI) for that entry of $\beta$. 

The procedure we follow is as described in Algorithm \ref{SimulationAlgorithm}. We use point-wise CIs as these are often used by practitioners to determine the significance of coefficients. If the DP estimate is in the CI, one would likely make the same inference using the DP estimate as the MLE. 

 For our simulations, we set $p=5$, $n=10^4$ or $n=10^6$, and consider $\ep\in\{1/16,1/8,\ldots,2,4\}$.
The results of these simulations are in Figure \ref{SimulationFigure}, where the $x$-axis denotes varying values of $\ep$, and the $y$-axis measures the proportion of times the estimate $\hat \beta^*$ falls in the $95\%$-CI of $\beta$. From these plots, we see that $\ell_\infty$-mech can reach the performance of $\ell_1$-mech with about half the privacy budget. For instance, in the bottom plot of Figure \ref{SimulationFigure} the $\ell_1$-mech achieves a fraction of approximately .7 at $\ep=1/2$, whereas the $\ell_\infty$-mech achieves a similar utility at $\ep=1/4$. This means that choosing $\ell_\infty$ over $\ell_1$-mech  results in DP estimates much closer to the true $\beta$. On the other hand, the $\ell_\infty$-mech and the optimal $K$-mech (derived in Subsection \ref{s:linear}) perform very similarly. Note that increasing $n$ improves the performance of all methods, but does not change the relative performance of these methods. Additional simulations indicated that the relative performance of the methods is similar for $n=10^2,10^3,10^5$ as well.

\begin{algorithm}
\caption{Simulate  Confidence Coverage}
\scriptsize
INPUT: $p$, $n$, $\ep$

\begin{algorithmic}[1]
  \setlength\itemsep{0em}
  \STATE Set $\beta = (0,-1.5,\ldots,1.5)\in \RR^{p+1}$, where the last $p$ entries are equally spaced.
\FOR{each of 200 replicates}
\STATE Draw $X_{ij}^0 \iid U[-1,1]$ for $i=1,\ldots, n$ and $j=1,\ldots,p$
\STATE Set $X = [1_n,X^0]$ and Draw $Y\sim N(X\beta, I_n)$
\ENDFOR
\FOR{ each replicate $(X,Y)$}
\STATE Compute $95\%$ CI for each of the last $p$ entries of $\beta$, based on $\hat \beta_{MLE}$
\ENDFOR
\FOR{each DP method and each replicate $(X,Y)$}
\STATE Compute the DP estimate $\beta_{DP}$ via \eqref{DPEstimate}
\STATE Compute average coverage :\\$C_{DP} = \frac{1}{p}\sum_{i=1}^p \#(\text{entries of $\beta_{DP}$ in its CI})$.
\ENDFOR
\end{algorithmic}
OUTPUT: $\frac{1}{200}\sum_{\text{replicates}}C_{DP}$ for each method of DP.
\label{SimulationAlgorithm}
\end{algorithm}

\begin{figure}[!h]
\begin{center}
\includegraphics[width=.7\textwidth]{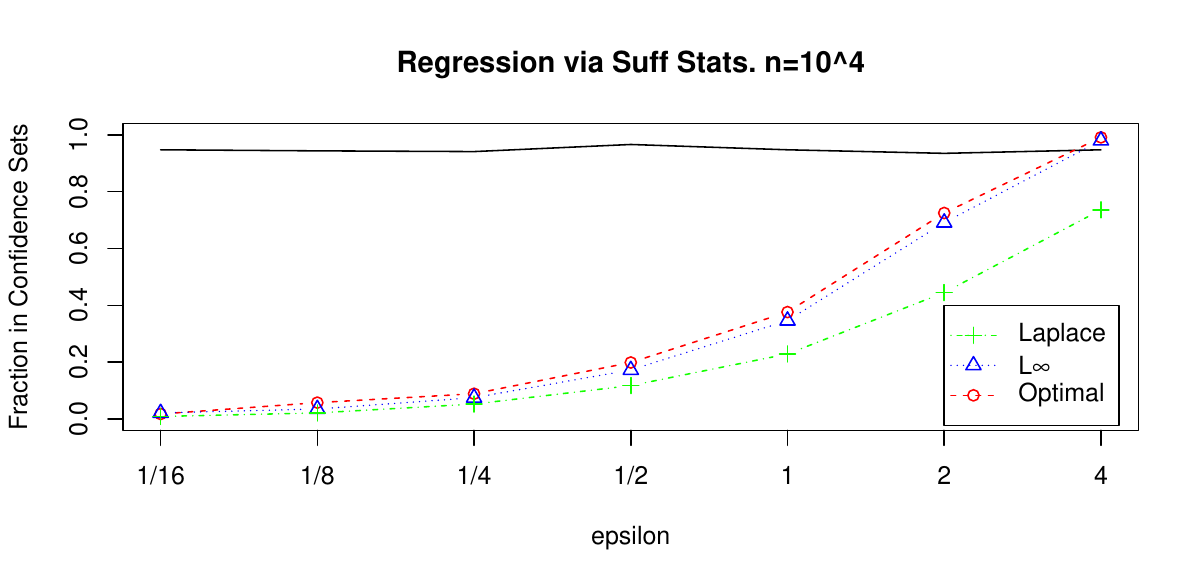}
\includegraphics[width=.7\textwidth]{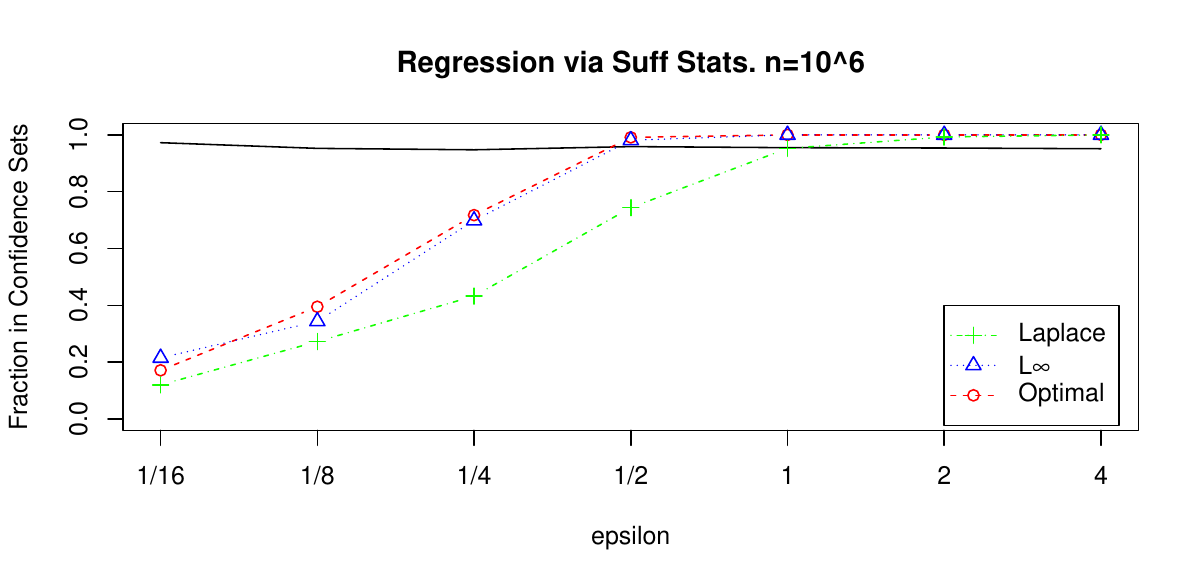}
\end{center}
\caption{Comparison of $\ell_1$-mech, $\ell_\infty$-mech, and optimal $K$-mech for linear regression, via Algorithm \ref{SimulationAlgorithm}. The estimates used in the above plot are via \eqref{DPEstimate}. For all simulations, $p=5$ and 200 replicates are used. In the top plot, $n=10^4$ and in the bottom plot, $n=10^6$. The solid line is how often the true $\beta$ falls in the confidence intervals, which is $\approx .95$. }
\label{SimulationFigure}
\end{figure}

\subsection{Linear Regression on Housing Data}\label{s:housing}
In this section we analyze a dataset containing information on $348,189$ houses in the San Francisco Bay area, collected between 2003 and 2006. Our response is rent, and the predictors are 
lot square-footage,
base square-footage,
 location in latitude and longitude,
 time of transaction,
 age of house,
 number of bedrooms,
and five indicators for the counties: Alameda, Contra Costa, Marin \& San Francisco
\& San Mateo, Napa \& Sonoma, Santa Clara. 

To clean the data, we follow a similar procedure  as in \citet{Lei2011:DPMest} and \citet{Lei2016}. We remove houses with prices outside of the range $105$ to $905$ thousand dollars, as well as houses with square-footage larger than $3000$.
In total, we have one response, 12 predictors, and 235,760 observations. As additional pre-processing, we apply a log-transformation to rent and both measures of square-footage. We then truncate all variables between the $0.0001$ and $.9999$ quantiles and then linearly transform the truncated variables to lie in $[-1,1]$. This procedure results in well-distributed values in each attribute. We also found that after this pre-processing, the assumptions of the linear model were reasonable.

As described in Subsection \ref{s:linear}, we  form the vector $T$ based on this data, add to it noise from either $\ell_1$-mech, $\ell_\infty$-mech, or the optimal $K$-mech, and post-process $T$ to get an estimate of the coefficient vector $\beta$ via \eqref{DPEstimate}.

We measure the performance of each DP estimate $\beta_{DP}$ by its $\ell_2$ distance to the MLE estimate $\beta_{MLE}$: $\lVert \beta_{DP} - \beta_{MLE}\rVert_2$. 
We give plots of the performance of $\ell_1$-mech, $\ell_\infty$-mech, and the optimal $K$-mech under this measure in Figure \ref{L2}. Each curve is an aggregate over 1000 replications of the DP algorithm. 
We see in Figure \ref{L2} that choosing $\ell_\infty$-mech over $\ell_1$-mech can affect performance about as much as doubling the privacy budget $\ep$. For a fixed $\ep$, the estimates from $\ell_\infty$-mech are considerably closer to the MLE and give substantially better estimates than the estimates from $\ell_1$-mech. On the other hand, the optimal $K$-mech does not offer sizeable benefits over $\ell_\infty$-mech.
\begin{figure}[!h]
\begin{center}
\includegraphics[width=.7\linewidth]{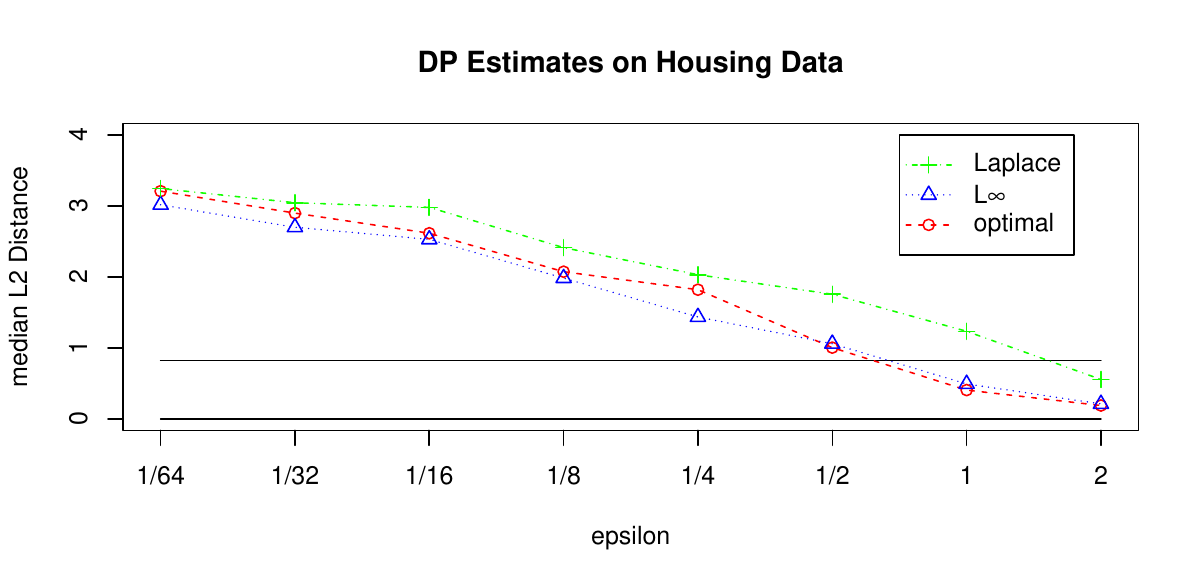}
\end{center}
\caption{Comparison of $\ell_1$-mech, $\ell_\infty$-mech, and optimal $K$-mech, calculated via \eqref{DPEstimate} for linear regression on Housing Data, measured by $\ell_2$ distance to $\hat \beta$. The solid line at height $\approx .82$ is the $\ell_2$ distance between the zero vector and $\hat \beta$. The lower horizontal line is at height $0$. 
  For each $\ep$,  1000 replicates are aggregated for each DP mechanism.}
\label{L2}
\end{figure}

\begin{remark}
The $\ell_1$ and $\ell_2$ sensitivities used in this example are not exact, and it may be possible to improve the performance of the $\ell_1$ and $\ell_2$ mechanisms slightly by optimizing these sensitivity calculations. Furthermore, the sensitivity space developed in Subsection \ref{s:linear}  assumed that all coordinates in $X$ can take values between $-1$ and $1$. However, in this example we have indicator variables which are inherently dependent (i.e. if the Alameda indicator variable is active, then the Contra Costa indicator must be inactive). This additional structure implies that the sensitivity space is actually smaller than the generic one developed in Subsection \ref{s:linear}. So, while all of our mechanisms in this section are valid, it is possible that they could be even further improved by taking these observations into account. 
\end{remark}

\section{Discussion}
\label{s:conclusions}


In this paper, we address the problem of releasing a noisy real-valued statistic vector $T$, a function of sensitive data under DP, via the class of $K$-norm mechanisms with the goal of minimizing the noise added to achieve privacy, and optimizing the use of the privacy-loss budget $\ep$. We propose a new notion we refer to as the {\em sensitivity space} to understand the geometric relation between a statistic and its sensitivity. We used the sensitivity space to study the class of $K$-mechs, a natural extension of the Laplace mechanism. Rather than naively using iid Laplace (the $\ell_1$-mech) or any other $K$-mech, we recommend choosing the $K$-mech based on properties of the  sensitivity space. To this end we propose three methods of evaluating the $K$-norm mechanisms in order to identify the optimal one for a fixed arbitrary (linear or non-linear) statistic $T$ and sample size $n$. We then showed a result fundamental for designing improved differentially private mechanisms: that the convex hull of the sensitivity space results in the optimal $K$-mech, which is stochastically tightest, has minimum entropy, and minimizes the conditional variance. On the other hand, if two (or more) $K$-mechs in particular are to be compared, this can be done by either checking for the containment of their associated norm balls, or by comparing the volume of their norm balls. In this case, even if using the convex hull is computationally intractible, we offer a framework  as illustrated in Figure \ref{fig:diagram}, which results in simple decision criteria to choose between several candidate $K$-mechs.

The proposed volume and containment orderings could be of broader statistical interest outside of privacy. In fact, we show that these orderings are connected to \emph{more scattered} and \emph{more dispersed}, which extend the notion of stochastic dominance to multivariate settings \citet{Zuo2000Nonparam}. Furthermore, as we showed in Section \ref{s:depth} stochastic tightness is closely related to stochastic depths, and may be of interest in that field as well.

Our extensions of objective perturbation and functional mechanism are also significant in the broader differential privacy community. Our modifications emphasize the flexibility that these mechanisms have in tailoring the noise-adding distribution to the problem at hand, and we show that by using the comparison criteria to determine the optimal $K$-mech, we are able to improve the performance of the output of these mechanisms, in terms of statistical utility. To facilitate the ease of implementing our proposed methods, we  provide a method of  sampling arbitrary $K$-mechs in Subsection \ref{s:sampling}, via rejecting sampling, and collect computationally efficient algorithms to sample the $\ell_1$, $\ell_2$, and $\ell_\infty$ mechanisms from the literature as well.

While the focus of this paper has been on reducing the noise introduced at a fixed level of $\ep$, alternatively our techniques allow one to achieve a desired amount of accuracy with a reduced privacy level $\ep$. This saves more of the privacy-loss budget for the computation of additional statistics or other statistical tasks, improving the usability of differentially private techniques. In fact, through the applications of linear and logistic regression, via the functional and objective perturbation mechanisms, we showed that choosing the $K$-mech based on our proposed criteria allows the same accuracy to be achieved with about half of the privacy-loss budget. The question of determining and setting the privacy-loss budget has come to prominence more recently with advances in DP methodologies and tools, and implementation of formal privacy in large organizations, especially those who must share data more broadly for an array of potential statistical analyses and maintain confidentiality, such as the U.S. Census. Theoretical results and practically implementable solutions such as those presented in this paper help address this fundamental question,

In Section \ref{s:compare}, we developed several methods of comparing $K$-norm mechanisms, and in Sections \ref{s:logisticSimulation}, \ref{s:linearSimulations}, and \ref{s:housing} we demonstrated that the proposed criteria provide a substantial improvement in practical utility. However, this may not always be the case. 
In particular, we provide an example below illustrating that even when the norm ball $H$ is contained in $K$, this does not imply that the $H$-mechanism outperforms the $K$-mechanism in terms of the marginal variance or the expected value of the $\ell_1$, $\ell_2$, or $\ell_\infty$ loss. 

{\bf A cautionary Example } Consider the example where $K$, the convex hull of the sensitivity space, is the blue solid rectangle in Figure \ref{fig:diamond.pdf}, and $H$ is the red textured diamond which contains $K$. We assume that the sensitivity of both is one. Then Theorem \ref{thm:variance} applies in this setting, and we know that the $K$-norm mechanism has smaller conditional variance than the conditional variance of the $H$-norm mechanism in every direction. However, the marginal variance of $H$ in the $x$ and $y$ coordinates are approximately $24139.87$ and   $242.09$ respectively, whereas the marginal variance of $K$ in the same directions are $40068.37$ and   $3.99$ using 100,000 samples from both mechanisms. In Table \ref{t:loss}, we see that the expected $L_\infty$, $L_1$ and $L_2$ loss are all minimized by $H$ rather than $K$. This example demonstrates a ``Simpson's paradox'' \citep{Blyth1972}, where the behavior of conditional variables is very different from the marginal variables

One reason this phenomena occurs is that the scaling in the $x$ and $y$ directions are of very different magnitudes. In terms of the marginal variance, the $K$-mech has a relatively large reduction in the $y$ direction, but a relatively moderate increase in the $x$ direction. However, in absolute terms, the increase in the $x$ direction dwarfs the reduction in the $y$-direction.

This same phenomenon could be constructed by instead setting the scales of $x$ and $y$ more equally, but using a loss function that disproportionately penalizes variability in the $x$ direction. Thus, it is important to note that in multivariate settings, without knowing the scaling of the variables and the particular loss function of interest, we cannot guarantee that any $K$-mech will outperform another in terms of these particular metrics.

On the other hand by scaling each entry of the statistic $T$ equally, as we did in Sections \ref{s:logistic} and \ref{s:linear} we are able to mitigate this problem as we saw in our numerical examples. 

\fig{width=.5\linewidth}{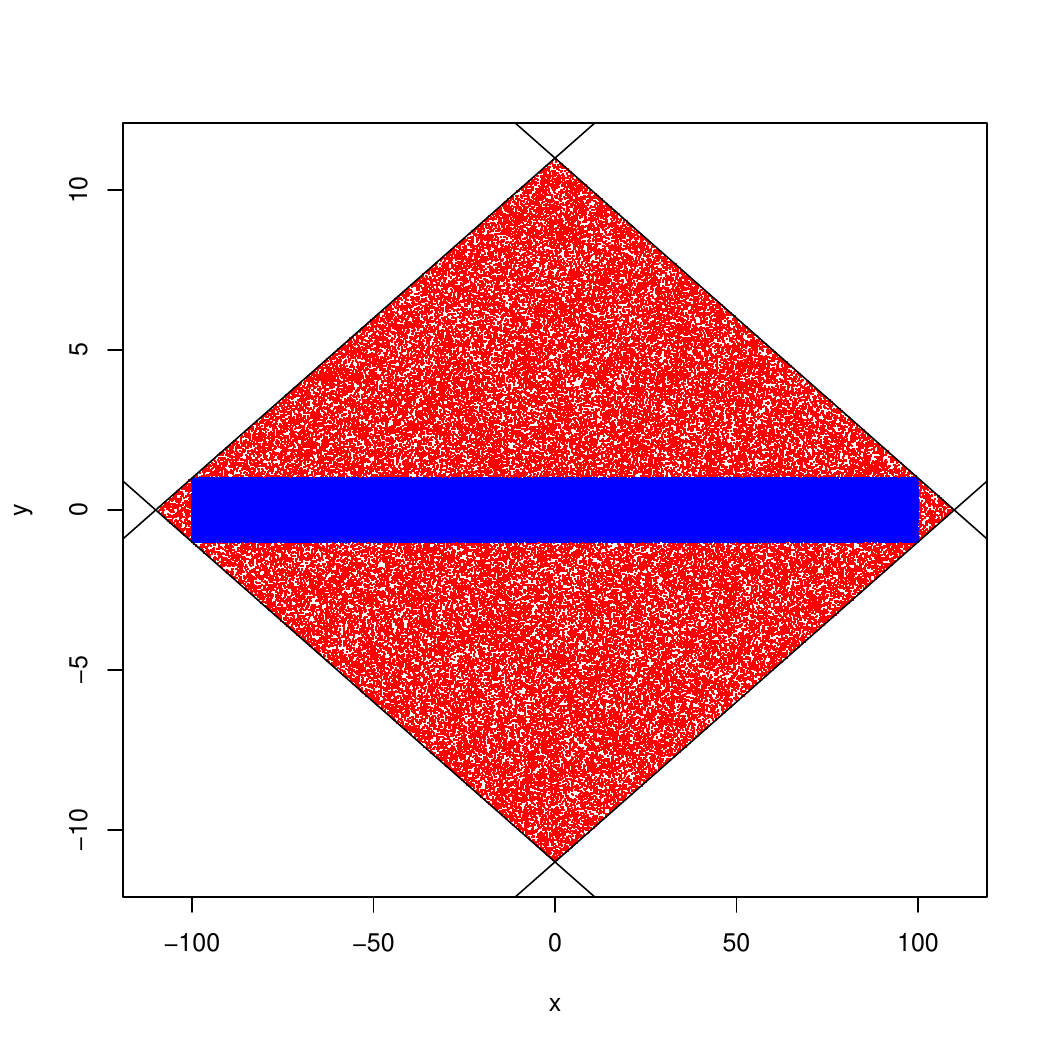}{The solid blue rectangle represents the space $K$ and the textured red diamond represents the space $H$. Note that $K$ is entirely contained in $H$.}

\begin{table}[H]
  \label{t:loss}
  \caption{Expected loss of the $K$ and $H$-norm mechanisms, described in Section 2. Monte carlo standard errors are in parentheses.}
  \centering
  \begin{tabular}{c|ccc}
    & $\ell_\infty$&$\ell_2$&$\ell_1$\\\hline
    $K$&150.17(0.418)&150.19(0.418)&151.66(0.420)\\
    $H$&110.89(0.345)&112.27(0.343)&120.89(0.349)
  \end{tabular}
\end{table}

{\bf Future work} As the Laplace mechanism is a part of many other mechanisms (i.e., Stochastic Gradient Descent \citep{Song2013:StochasticGradient} and Subsample-Aggregate \citep{Smith2011:Privacy-preservingSE}), our methodology can be used to improve finite sample performance of other mechanisms as well. Recently, \citet{Reimherr2019} developed a new mechanism, which can be viewed as a hybrid of the exponential mechanism and objective perturbation, which they call the $K$-norm gradient mechanism (KNG). Using the sensitivity space and tools of this paper, we may be able to optimize the performance of KNG, by choosing the optimal norm. Optimizing the performance of these mechanisms allows for improved statistical inference, increased usability of DP methods, and better use of the privacy-loss budget.



\bibliographystyle{plainnat}
\bibliography{./DataPrivacyBib.bib}{}

\section{Appendix }
\label{s:appendix}
\subsection{Implementing $K$-Norm Mechanisms}
\label{s:sampling}

\begin{algorithm}
\caption{Sampling from $\ell_1$-mech}
\scriptsize
INPUT: $T(X)$, $\Delta_1(T)$, and $\ep$
\begin{algorithmic}[1]
  \STATE Set $m\defeq {\rm length}(T(X))$.
\STATE Draw $V_j \iid \mathrm{Laplace}((\frac{\ep}{\Delta_1(T)})^{-1})$ for $j=1,\ldots,m$
\STATE Set $V = (V_1,\ldots, V_m)^\top$
\end{algorithmic}
OUTPUT: $T(X)+V$
\label{SampleL1}
\end{algorithm}

\begin{algorithm}
\caption{Sampling from $\ell_2$-mech \citep{Yu2014}}
\scriptsize
INPUT: $T(X)$, $\Delta_2(T)$, and $\ep$
\begin{algorithmic}[1]
  \STATE Set $m\defeq {\rm length}(T(X))$.
\STATE Draw $Z\sim N(0,I_m)$
\STATE Draw $r \sim \mathrm{Gamma}(\al = m,\beta = \ep/\Delta_2(T))$
\STATE Set $V = \frac{rZ}{\lVert Z\rVert_2}$
\end{algorithmic}
OUTPUT: $T(X)+V$
\label{SampleL2}
\end{algorithm}

In this section, we review algorithms to implement the $\ell_1,\ell_2,\ell_\infty$-mechs. Then we give a  method to implement arbitrary $K$-mechs. The $\ell_1$-mech can be easily implemented via Algorithm \ref{SampleL1}, which only uses independent Laplace random variables. Algorithm \ref{SampleL2}, which appears in \citet{Yu2014}, gives a method to sample the the $\ell_2$-mech. 
Algorithm \ref{SampleLinfty} appears in \citet{Steinke2017} and gives a method to sample the $\ell_\infty$-mech.

\begin{algorithm}
\caption{Sampling from $\ell_\infty$-mech   \citep{Steinke2017}}
\scriptsize
INPUT: $T(X)$, $\Delta_2(T)$, and $\ep$
\begin{algorithmic}[1]
  \STATE Set $m\defeq {\rm length}(T(X))$.
\STATE Set $U_j \iid U(-1,1)$ for $j=1,\ldots, m$
\STATE Draw $r \sim \mathrm{Gamma}(\al = m+1,\beta = \ep/\Delta_\infty(T))$
\STATE Set $V = r\cdot (U_1,\ldots, U_m)^\top$
\end{algorithmic}
OUTPUT: $T(X)+V$
\label{SampleLinfty}
\end{algorithm}

In general, sampling from the $K$-Norm mechanisms is non-trivial. \citet[Remark 4.2]{Hardt2010:GeometryDP}
 gives a method of sampling from $K$-mech, provided that one can determine whether a point is in $K$. Precisely, we require a function $I_{K}:\RR^m \rightarrow\{0,1\}$
given by $I_K(u) = 1$ if $u\in K$ and $I_K(u)=0$ otherwise. 
 in Algorithm \ref{SampleK}, we propose a  procedure to sample from $K$ using rejection sampling (see \citet[Chapter 11]{Bishop2006} for an introduction).

\begin{algorithm}
\caption{Sampling from $K$-Norm Mechanism with Rejection Sampling}
\label{RejectionAlgorithm}
\scriptsize
INPUT: $\ep$, $\Delta_\infty(T)$, $\Delta_K(T)$, $I_K(\cdot)$, $T(X)$
\begin{algorithmic}[1]
\STATE Set $m = \mathrm{length} (T(X))$.
\STATE Draw $r \sim \mathrm{Gamma}(\al = m+1, \beta = \ep/\Delta_K(T))$
\STATE Draw $U_j \iid \mathrm{Uniform}(-\Delta_\infty(T), \Delta_\infty(T))$ for $j=1,\ldots, m$
\STATE Set $U = (U_1,\ldots, U_m)^\top$
\STATE If $I_K(U) = 1$, release $T(X) + r\cdot U$, else go to 3).
\end{algorithmic}
\label{SampleK}
\end{algorithm}

\begin{example}
Back to the setting of Example \ref{SquaredExample}. The space $K_T = \mathrm{span}(S_T)$ is
\[K_T = \l\{(u_1,u_2)\in [-2,2]^2\middle | |u_2| \leq \begin{cases}
2-(1-u_1)^2&\text{if } u_1\geq 1\\
2-(u_1+1)^2&\text{if } u1<-1\end{cases}\r\}\]
Then for a vector $u = (u_1,u_2)\in \RR^2$, our indicator function is 
$I_{K_T}(u) = 1$ if $u\in K_T$ and $I_{K_T}(u)=0$ otherwise.
\end{example}

In \citet{Hardt2010:GeometryDP}, they propose a random grid-walk procedure to sample from $K$. However, this only gives approximate sampling. On the other hand, Algorithm \ref{SampleK} is easily implemented, and gives exact uniform sampling from $K$.

\subsection{Proofs and Technical Lemmas}
\label{s:proofs}



\begin{lem}\label{lem:Lebesgue}
  Let $\mscr S$ be a colletion of (Lebesgue) measurable sets in $\RR^m$ such that $\arg\inf_{S\in \mscr S} \lambda(S)$ is unique, where $\la(\cdot)$ is Lebesgue measure. Let $A:\RR^m\rightarrow \RR^m$ be an invertible linear transformation. Then $\arg\inf_{T\in A\mscr S} \lambda(T) = A \left(\arg\inf_{S\in \mscr S} \lambda(S)\right)$.
\end{lem}
\begin{proof}
  First note that for any measurable $S$,
  \[ \lambda(AS) = \int_{AS} 1 \ dx = \int_S |\mathrm{det}(A^{-1})| \ du 
    = |\mathrm{det}(A^{-1})| \lambda(S),\]
  where we use the change of variables formula. Thus,
  \[\arg\inf_{T\in A\mscr S} \lambda(T) = A\arg\inf_{S\in \mscr S} \lambda(AS)
    =A \arg\inf _{S\in \mscr S} |\mathrm{det}(A^{-1})| \lambda(S)
    =A\arg\inf_{S\in \mscr S} \lambda(S).\qedhere\]
\end{proof}

\begin{lem}\label{lem:Sprops}
  Let $X$ be a random variable on $\RR^m$ which is unimodal (center zero), continuous, and decreasing away from the center (i.e. $f_X(x)\leq f_X(ax)$ for $a\in [0,1]$), and such that for all $t>0$, $P(\{x\mid f_X(x) = t\})=0$. Then
  \begin{enumerate}
  \item $S_X^\alpha$ is unique,
  \item $S_X^\alpha\subset S_X^\beta$ for $\alpha\leq \beta$ (nested),
  \item $c S_X^\alpha \subset S_X^\alpha$ for $c\in [0,1]$ (linear closure wrt center),
  \item $S_{Ax}^\alpha = A S_X^\alpha$ for any linear, invertible map $A:\RR^m \rightarrow \RR^m$.
  \end{enumerate}
\end{lem}
\begin{proof}
  \begin{enumerate}
  \item Since $X$ is continuous and the sets $\{x\mid f_X(x)=t\}$ have probability zero for all $t>0$, there exists $C(\alpha)$ such that $S_X^\alpha = \{x\mid f_X(x) \geq C(\alpha)\}$, where $C(\alpha)$ is a decreasing function of $\alpha$ \citep[Theorem 9.3.2]{Casella2002}. This establishes uniqueness.
  \item Let $x\in S_X^\alpha$. Then $f_X(x) \geq C(\alpha)\geq C(\beta)$. Hence, $x\in S_X^\beta$.
  \item We calculate \[cS_X^\alpha = \{cx\mid f_X(x) \geq C(\alpha)\}
      =\{y\mid f_X(c^{-1} y) \geq C(\alpha)\}
      \subset \{y\mid f_X(y) \geq C(\alpha)\}
      =S_X^\alpha,\]
    where we use the fact that $f_X(a^{-1}y)\leq f_X(y)$.
  \item Define $\mscr S_X^\alpha = \{S\mid P(X\in S) \geq \alpha\}$. Then $S_X^\alpha = \arg\inf_{S\in \mscr S} \lambda(S)$, where $\la(\cdot)$ is Lebesgue measure. Then
    \[A \mscr S_X^\alpha = \{AS \mid P(X\in S) \geq \alpha\}
      =\{T\mid P(X\in A^{-1} T)\geq \alpha\}
      =\{T\mid P(AX\in T) \geq \alpha\}
      =\mscr S_{AX}^\alpha.\]
    Taking the infinum with respect to Lebesgue measure on both sides yields $AS_X^\alpha = S_{AX}^\alpha$. We are able to pass $A$ in front of the infimum by Lemma \ref{lem:Lebesgue}.\qedhere
  \end{enumerate}
\end{proof}

\begin{proof}[Proof of Lemma \ref{lem:gamma}.]
  We will compute the moment generating function (MGF) of $\lVert V \rVert_K$. Let $0\leq t\leq a$. Call $\alpha$ the integrating constant. Then 
  \begin{align*}
    \EE \exp\l(t\lVert V\rVert_K\r)& = \alpha^{-1} \int\cdots \int \exp(t\lVert v \rVert_K) \exp(-a\lVert v \rVert_K) \ dv_1\ldots, dv_m\\
                                   &=\alpha^{-1} \int\cdots \int \exp(-(a-t) \rVert v \rVert_K) \ dv_1, \ldots, dv_m\\
                                   &= (a-t)^{-m} a^m\\
    &=(1-t/a)^{-m},
  \end{align*}
  where we applied a $u$-substitution, noting that the integrand is of the same form as $f_V$. We identify this as the MGF of the random variable $\mathrm{Gamma}(m,a)$. 
\end{proof}

\begin{proof}[Proof of Lemma \ref{lem:concentration}.]
  First note that $S_V^\alpha$ is of the form $S_V^\alpha = \{v\mid f_V(v)\geq t\}$ for some $t$ \citep[Theorem 9.3.2]{Casella2002}. Since $f_V(v)$ is an increasing function of $\lVert v \rVert_K$, equivalently, we have $S_V^\alpha  = \{v\mid \lVert v \rVert_K \geq t\}$. We must determine the value of $t$ such that $P(V\in\{v\mid \rVert v \rVert_K\geq t\}) = \alpha$. Recall from Lemma \ref{lem:gamma} that $\lVert V\rVert_K \sim \mathrm{Gamma}(m,a)$. We conclude that $t$ is the $\alpha$ quantile of $\mathrm{Gamma}(m,a)$.
\end{proof}

\begin{proof}[Proof of Theorem \ref{thm:tightness}.]
  By the previous lemma, we know that for all $\alpha\in (0,1)$, $S_V^\alpha = \{x \mid \lVert x \rVert_{K_V}\leq t_\alpha\} = t_\alpha\cdot K_V$ and $S_W^\alpha = \{x \mid \lVert x \rVert_{K_W}\leq t_\alpha\} = t_\alpha\cdot K_W$  for the same value of $t$. Since $K_V\subset K_W$, we have that $S_V^\alpha \subset S_W^\alpha$ for all $\alpha \in (0,1)$. 
\end{proof}

\begin{proof}[Proof of Theorem \ref{thm:depth}]
   All we have to show is that $D_X$ is a depth function for unimodal, continuous, decreasing random variables. For simplicity, we assume that the center of $X$ is zero.
  \begin{itemize}
  \item [(A1)]
    \begin{align*}
      D_{AX}(Ax)&= 1-\inf \{\alpha \mid Ax \in S_{Ax}^\alpha\}\\
                &= 1- \inf \{ \alpha \mid x\in A^{-1} S_{Ax}^{\alpha}\}\\
                &=1-\inf \{\alpha \mid X\in A^{-1} A S_X^{\alpha}\}\\
                &= 1-\inf \{\alpha \mid x\in S_{X}^\alpha\}\\
      &=D_X(x),
    \end{align*}
    where we use property 4) of Lemma \ref{lem:Sprops}.
  \item [(A2)] By property 2) of Lemma \ref{lem:Sprops}, we know that the sets $S_X^\alpha$ are nested. So, the minimum value of $D_X$ is attained at the mode, which has depth of $1$.
  \item [(A3)] Since we assume that the center is at $x_0=0$, it suffices to show that $D_X(x) \leq D_X(ax)$ for $a\in[0,1]$. Let $y\in \RR^m$. First we will show that $\{\alpha \mid y\in a S_X^\alpha\} \subset \{\alpha \mid y\in S_X^\alpha\}$. Let $\alpha\in \{\alpha\mid y\in a S_X^\alpha\}$. Then $y\in a S_X^\alpha \subset S_X^\alpha$, by property 3) of Lemma \ref{lem:Sprops}. So, $\alpha \in \{\alpha\mid y\in S_X^\alpha\}$. It follows that $\inf \{\alpha \mid y\in a S_X^\alpha\}\geq \inf \{ \alpha\mid y\in S_X^\alpha\}$. Finally, by choosing $y=ax$ we have that
    \[D_X(x) = D_{aX}(y)
      =1-\inf\{\alpha \mid y\in a S_X^\alpha\}
      \leq 1-\inf\{\alpha \mid y\in S_X^\alpha\}
      =D_X(y)=D_X(ax),\]
    where we use property (A1) of Definition \ref{def:depth}, and property 4) of Lemma \ref{lem:Sprops}.
  \item [(A4)] We will show that $\lim_{t\rightarrow \infty} D_X(tx)=0$ for all $x\in \RR^m$. Let $\ga>0$ and $x\in \RR^m$ be given. By assumption, we have that for all $x$, $\lim_{t\rightarrow \infty} f_X(tx)=0$. So, there exists $t>0$ such that $f(tx)\leq C(1-\ga)$, where $C(\cdot)$ is the function defined in the proof of Lemma \ref{lem:Sprops} part 1. Then $D_X(tx) = 1-\inf\{\alpha\mid f_X(tx)\geq C(\alpha)\}\leq 1-(1-\ga)=\ga$, where we use the fact that $1-\ga\leq \inf \{\alpha\mid f_X(tx)\geq C(\alpha)\}$.
  \end{itemize}
\end{proof}

\begin{proof}[Proof of Proposition \ref{prop:entropy}.]
  We can write the integrand as $-\log f(V) = -\log\left((\ep/\Delta)^m \frac{1}{m! \la(K)}\right) + \frac{\ep}{\Delta} \lVert V \rVert_K$. Then
  \begin{align*}
    H(V)&= \EE-\log f(V)\\
        &= \log\left( (\Delta/\ep)^m m! \la(K)\right) + \frac{\ep}{\Delta} \EE \lVert V \rVert_K\\
    &=\log\left( (\Delta/\ep)^m m! \la(K)\right) + \frac{\ep}{\Delta} \frac{m \Delta}{\ep},
  \end{align*}
  where we recall from Lemma \ref{lem:gamma} that $\lVert V \rVert_K \sim \mathrm{Gamma}(m,\ep/\Delta)$.
\end{proof}

\begin{proof}[Proof of Lemma \ref{lem:conditional}.]
\begin{enumerate}
    \item Recall from  \citet[Remark 4.2]{Hardt2010:GeometryDP} that $V \overset d =R\cdot U$, where $R\sim \mathrm{Gamma}(m+1,a)$, $U\sim \mathrm{Unif}(K)$, and $R\indep U$. Then $\lVert V \rVert_K = R \lVert U \rVert$, and $\frac{V}{\lVert V \rVert_K} = \frac{RU}{R\lVert U \rVert_K} = \frac{U}{\lVert U \rVert_K}$. So, it suffices to show that $U \indep \frac{U}{\lVert U \rVert_K}$. To this end, we will derive the conditional cdf of $\lVert U \rVert_K$ given that $\frac{U}{\lVert U \rVert_K}=e$ for some $\lVert e \rVert_K=1$. 
    
    First, for any $\ga>0$, we define the set $\mathrm{Cone}(e,\ga) = \left\{v\in \RR^m \mid \arccos\left( \frac{v^\top e}{\lVert v\rVert_2 \cdot \lVert e \rVert_2}\right) \leq \ga\right\}$, which is the set of all vectors $v$ whose angle from $e$ is less than $\ga$. Then for any $t\geq 0$, we can express the conditional cdf as 
    \begin{align*}
        P\left(\lVert U \rVert_K \leq t \ \Big| \ \frac{U}{\lVert U \rVert_K}=e\right)
        &=\lim_{\ga\rightarrow 0} \frac{P(U \in \left( tK \cap \mathrm{Cone}(e,\ga)\right))}{P(U \in (K\cap \mathrm{Cone}(e,\ga)))}\\
        &=\lim_{\ga\rightarrow 0} \frac{\lambda(tK\cap \mathrm{Cone}(e,\ga))/\lambda(K)}{\lambda(K\cap \mathrm{Cone}(e,\ga))/\lambda(K)}\\
        &=\lim_{\ga\rightarrow 0} \frac{t^m \lambda(K \cap \mathrm{Cone}(e,\ga))}{\lambda(K\cap \mathrm{Cone}(e,\ga))}\\
        &= t^m.
    \end{align*}
    We see that the conditional cdf $P(\lVert U \rVert_K \leq t\mid \frac{U}{\lVert U \rVert_K} = e)$ does not depend on $e$. We conclude that $\lVert U \rVert_K \indep \frac{U}{\lVert U \rVert_K}$ and hence $\lVert V \rVert_K \indep \frac{V}{\lVert V\rVert_K}$.
    \item Since $V\in \mathrm{span}(e)$, we know that $|V^\top e| = \lVert V\rVert_2 \cdot \lVert e \rVert_2=\lVert V \rVert_2$. Then 
    \begin{align*}
    |V^\top e| &= \lVert V \rVert_2 
    = \left \lVert \lVert V \rVert_K \cdot \frac{V}{\lVert V \rVert_K}\right\lVert_2
    =\lVert V\rVert_K \left \lVert  \frac{e}{\lVert e \rVert_K}\right\lVert_2
    =\frac{\lVert V \rVert_K}{\lVert e \rVert_K}
    \sim \mathrm{Gamma}(m,a\lVert e \rVert_K),
    \end{align*}
    where we use the fact that $\frac{V}{\lVert V \rVert_K} =\pm \frac{e}{\lVert e \rVert_K}$, that 
    $\lVert V\rVert_K$ is independent of $\frac{V}{\lVert V \rVert_K}$, and that $\lVert V \rVert_K \sim \mathrm{Gamma}(m,a)$ from Lemma \ref{lem:gamma}.\qedhere
\end{enumerate}
\end{proof}

\begin{proof}[Proof of Theorem \ref{thm:variance}.]
By Lemma \ref{lem:conditional}, $W=\l(V_K^\top e\middle | V_K\in E\r)$ is distributed as $\mathrm{Gamma}(m,\frac{\ep\lVert e \rVert_K}{\Delta_K})$, which has variance $\frac{m\Delta_K}{\ep^2\lVert e \rVert_K}$. Minimizing the variance between the $K$-norm and $H$-norm is equivalent to maximizing $\frac{\lVert e \rVert_K}{\Delta_K}$. Note that $\frac{\lVert e \rVert_K}{\Delta_K}$ is the same value as the diameter of the set $\mathrm{span}(e)\cap (\Delta_K\cdot K)$. Since $\Delta_K\cdot K \subset \Delta_H\cdot H$, we have that $\mathrm{span}(e)\cap (\Delta_K\cdot K)\subset \mathrm{span}(e)\cap (\Delta_H\cdot H)$. The result follows.
\end{proof}

\begin{proof}[Proof of Theorem \ref{ThmObjPert}.]
First we will assume that $r(\theta)$ is twice-differentiable and $\Theta = \RR^m$. At the end of the proof we will use the results in \citet{Kifer2012:PrivateCERM} to extend the proof to arbitrary convex $r(\theta)$ and arbitrary convex sets $\Theta$. 

By \citet[Proposition 2.3]{Awan2019}, it suffices to show that for all $a\in \RR^m$ and all $X,X'\in \mscr X^n$ with $\de(X,X')=1$ we have 
\[\frac{\mathrm{pdf}(\ta_{DP} = a\mid X)}{\mathrm{pdf}(\ta_{DP} = a\mid X')} \leq \exp({\ep}).\]
To this end, let $a\in \RR^m$ and $X,X'$ be such that $\de(X,X') = 1$. Then if $\ta_{DP}=a$, we have
$\ds a = \arg\min_{\ta\in \RR^m} n \hat {\mscr L} (a; X) + r(\ta) + \frac{\ga}{2}\ta^\top \ta + V^\top \ta$. 
By taking the gradient with respect to $\ta$ and setting it equal to zero, we can solve for $V$ as a function of $a$: $V(a;X) =-(n\nabla \hat {\mscr L}(a;X) + \nabla r(a) + \ga a).$
Then applying this one-to-one change of variables, we get 
\[\frac{\mathrm{pdf}(\ta_{DP} = a\mid X)}{\mathrm{pdf}(\ta_{DP} = a\mid X')}
=\frac{f(V(a;X)\mid X)}{f(V(a;X')\mid X')}\frac{|\det \nabla V(a;X')|}{|\det\nabla V(a;X)|}\]
We will bound these two factors separately. First, we have 
$\frac{f(V;X)}{f(V;X')} \leq \exp({\ep q}),$
by Proposition \ref{prop:KNorm}.
As $\de(X,X')=1$, without loss of generality, assume that $X_i=X_i'$ for all $i=1,\ldots,n-1$. Call 
$\ds A= \nabla V(a;X),$ 
$B= \nabla V(a;X'),$ 
and $C= \sum_{i=1}^{n-1} \nabla^2 \ell(a; X_i) + \nabla^2 r(a) + \ga I_m,$
where $I_m$ is the $m\times m$ identity matrix. Note that $A = C + \nabla^2\ell(a; X_n)$ and $B =C+\nabla^2\ell(a;X_n')$. Then we have 
\begin{align}
\l| \frac{\det\nabla V(a;X')}{\det\nabla V(a;X)}\r|
&= \l| \frac{\det(B)}{\det(A)}\r|
=\l| \frac{\det(C+\nabla^2 \ell(a;X_n'))}{\det(C + \nabla^2\ell(a;X_n))}\r|\\
&= \l| \frac{\det(C) \det(I_m + C^{-1}\nabla^2 \ell(a;X_n'))}{\det(C) \det(I_m + C^{-1} \nabla^2 \ell(a;X_n))}\r|\\
&\leq \frac{1+ \frac{\la}{\ga}}{|\det(I_m + C^{-1}\nabla^2(a;X_n))|}\label{det1}\\
&\leq 1+\frac{\la}{\ga}\label{det2}= \exp({\ep(q-1)})
\end{align}
To justify the inequality in \eqref{det1}, note that $C^{-1}\nabla^2 \ell(a;X_n')$ is positive definite of rank at most $1$. We know that $\nabla^2 \ell(a;X_n')$ has at most one nonzero eigenvalue. Furthermore, $\ga$ is a lower bound on the eigenvalues of $C$,since $C$ is the sum of positive definite functions and $\ga I_m$. So, $C^{-1} \nabla^2\ell(a;X_n')$ has at most one nonzero eigenvalue, which is bounded between $0$ and $\la/\ga$. 

Next we  justify the inequality in \eqref{det2}. Since $C^{-1}\nabla^2 \ell(a;X_n)$ is positive definite, all of its eigenvalues are non-negative. Thus, all eigenvalues of $(I + C^{-1}\nabla^2\ell(a;X_n))$ are greater than or equal to $1$. Hence, $\det(I + C^{-1}\nabla^2\ell(a;X_n))\geq 1$.

The last equality just uses the fact that $\ga = \frac{\la}{\exp({\ep(q-1)})-1}$. 
Finally, we combine our bounds:
\[\frac{\mathrm{pdf}(\ta_{DP}=a;X)}{\mathrm{pdf}(\ta_{DP}=a;X')} = \frac{f(V;X)}{f(V;X')} \l| \frac{\det(\nabla V(a;X'))}{\det(\nabla V(a;X))}\r|\leq \exp({\ep q}) \exp({\ep(q-1)}) = \exp({\ep}).\]

Theorem 1 in \citet{Kifer2012:PrivateCERM} on successive approximations extends without modification to Definition \ref{DP}. Using this theorem along with the techniques detailed in Appendix C.2, C.3 of \citet{Kifer2012:PrivateCERM}, we extend this proof to arbitrary convex functions $r(\ta)$ and arbitrary convex sets $\Theta$. 
\end{proof}



\end{document}